\documentclass[11pt]{article}

\usepackage
{
        amssymb,
        amsthm,
        amsmath,
        fullpage,
        times
}

\DeclareMathOperator {\Var}  {Var}
\def\compactify{\itemsep=0pt \topsep=0pt \partopsep=0pt \parsep=0pt}

\newcommand {\symdiff} {\bigtriangleup}
\newcommand {\roundup}   [1] {{\lceil {#1} \rceil}}
\newcommand {\rounddown} [1] {{\lfloor {#1} \rfloor}}

\newcommand {\set}   [1] {\left\{ #1 \right\}}
\newcommand {\eps}       {\varepsilon}

\newcommand {\E}       {\mathbb{E}}

\newcommand {\bbZ}    {\mathbb{Z}}

\newcommand {\bbR}    {\mathbb{R}}
\newcommand {\bbN}    {\mathbb{N}}

\newcommand {\calN}   {{\cal{N}}}

\newcommand {\calI}   {{\cal{I}}}

\newtheorem{theorem}{Theorem}[section]
\newtheorem{lemma}[theorem]{Lemma}

\newtheorem{claim}[theorem]{Claim}
\newtheorem{corollary}[theorem]{Corollary}
\newtheorem{definition}[theorem]{Definition}

\newtheorem{remark}[theorem]{Remark}

\newcommand{\ThmStyle}[2]{\noindent\textbf{Theorem~\ref{#1}.} \textit{#2}}
\newcommand{\LemStyle}[2]{\noindent\textbf{Lemma~\ref{#1}.} \textit{#2}}

\newcommand{\DeclareThmA}[2]{%%
   \begin{theorem}\label{thm:#1}%%
   #2%%
   \end{theorem}%%
   \newcommand{\ThmA}{\ThmStyle{thm:#1}{#2}}%%
    }

\newcommand{\DeclareLemA}[2]{%%
   \begin{lemma}\label{lem:#1}%%
   #2%%
   \end{lemma}%%
   \newcommand{\LemA}{\LemStyle{lem:#1}{#2}}%%
    }

\newcommand{\DeclareLemB}[2]{%%
   \begin{lemma}\label{lem:#1}%%
   #2%%
   \end{lemma}%%
   \newcommand{\LemB}{\LemStyle{lem:#1}{#2}}%%
    }

\title{How to Play Unique Games against a Semi-Random Adversary\\ \small{$\phantom{.}$} \\ \large{Study of Semi-Random Models of Unique Games}}
\author{Alexandra Kolla\\Microsoft Research \and Konstantin Makarychev\\IBM Research \and
Yury Makarychev\\TTIC}
\date{}

\begin{document}
\maketitle
\begin{abstract}
In this paper, we study the average case complexity of the Unique Games problem.
We propose a natural semi-random model, in which a unique game instance is generated in several steps.
First an adversary selects a completely satisfiable instance of Unique Games, then
she chooses an $\eps$--fraction of all edges, and finally replaces (``corrupts'') the constraints corresponding
to these edges with new constraints. If all steps are adversarial, the adversary can obtain
any $(1-\eps)$ satisfiable instance, so then the problem is as hard as in the worst case.
In our semi-random model, one of the steps is random, and all other steps are adversarial.
We show that known algorithms for unique games (in particular, all algorithms that use the standard SDP relaxation)
fail to solve  semi-random instances of Unique Games.

We present an algorithm that with high probability finds a solution satisfying a $(1-\delta)$ fraction 
of all constraints in semi-random instances (we require that the average degree of the graph is $\tilde \Omega(\log k)$). 
To this end, we consider a new non-standard SDP program for Unique Games, which is not a relaxation for the problem,
and show how to analyze it. We present a new rounding scheme that simultaneously 
uses SDP and LP solutions, which we believe is of independent interest.

Our result holds only for $\eps$ less than some absolute constant. We prove that if $\eps \geq 1/2$, then the problem is hard in one of the models,
that is, no polynomial--time algorithm can distinguish between the following two cases: (i) the instance is a $(1-\eps)$ satisfiable
semi--random instance and (ii) the instance is at most $\delta$ satisfiable (for every $\delta > 0$);
the result assumes the $2$--to--$2$ conjecture.

Finally, we study semi-random instances of Unique Games that are \textit{at most} $(1-\eps)$ satisfiable.
We present an algorithm that with high probability, distinguishes between the case when the instance is a semi-random instance
and the case when the instance is an (arbitrary) $(1-\delta)$ satisfiable instance if $\eps > c\delta$ (for some absolute constant $c$).
\end{abstract}
\setcounter{page}{0} % so that title page is page 0
\thispagestyle{empty} %%suppress the first page number (which is "0")
\pagebreak

%%%%%%%%%%%%%%%%%%%%%%%%%%%%%%%%%%%%%%%%%%%%%%%%%%%%%%

\section{Introduction}
In this paper, we study the average case complexity of the Unique Games problem in a semi-random model.
In the Unique Games problem, we are given a graph $G=(V,E)$ (denote $n=|V|$), a set of labels $[k]=\{0,\dots, k-1\}$ and a set
of permutations $\pi_{uv}$ on $[k]$, one permutation for every edge $(u,v)$. Our goal is to assign a label (or state)
$x_u\in [k]$ to every vertex $u$ so as to maximize the number of satisfied constraints of the form $x_v = \pi_{uv}(x_u)$.
The value of the solution is the number of satisfied constraints.

The problem is conjectured to be very hard in the worst case. The Unique Games Conjecture (UGC) of Khot~\cite{Kho02}
states that for every positive $\eps$, $\delta$ and sufficiently large $k$, it is NP-hard to distinguish between
the case where at least a $1-\eps$ fraction of constraints is satisfiable, and the case where at most a $\delta$
fraction of all constraints is satisfiable.

 One reason which makes UGC particularly intriguing is its numerous implications. The conjecture, if true, implies that the currently best known approximation algorithms for
many important computational problems have optimal approximation ratios.
Indeed, since its origin, UGC has been successfully used to prove
often optimal hardness of approximation results for several
important NP-hard problems such as MAX CUT~\cite{KKMO05}, Vertex Cover~\cite{KR03},
Maximum Acyclic Subgraph~\cite{GMR}, Max $k$-CSP~~\cite{Rag, GR08, ST06}, which are not known to follow from standard complexity
assumptions.

 Arguably, a seemingly strong reason
for belief in UGC is the failure of several attempts to design efficient algorithms for Unique Games using current state-of-the-art techniques, even though a large amount of research activity in recent years has focused on the design of such algorithms. One direction of research has concentrated on
developing polynomial-time approximation algorithms for arbitrary instances of unique games. The first algorithm
was presented by Khot in his original paper on the Unique Games Conjecture~\cite{Kho02}, and then several algorithms
were developed in papers by Trevisan~\cite{Tre05}, Gupta and Talwar~\cite{GT}, Charikar, Makarychev and Makarychev~\cite{CMM1},
Chlamtac, Makarychev, and Makarychev~\cite{CMM2}. Another direction of research has been to study subexponential approximation
algorithms for Unique Games. The work was initiated by Kolla~\cite{Kolla} and Arora, Impagliazzo, Matthews and Steurer~\cite{AIMS}
who proposed subexponential algorithms for certain families of graphs.
Then, in a recent paper, Arora, Barak and Steurer~\cite{ABS} gave a subexponential algorithm for arbitrary instances of
Unique Games.

These papers, however, do not disprove the Unique Games Conjecture. Moreover, Khot and Vishnoi~\cite{KV} showed that it is impossible to
disprove the Conjecture by using the standard semidefinite programming relaxation for Unique Games, the technique used in the best
currently known polynomial-time approximation algorithms for general instances of Unique Games. Additionally,
Khot, Kindler, Mossel, and O'Donnell~\cite{KKMO05} proved that the approximation guarantees obtained in~\cite{CMM1}
cannot be improved if UGC is true (except possibly for lower order terms).

All that suggests that Unique Games is a very hard problem. Unlike many other problems, however, we do not know any specific
families of hard instances of Unique Games. In contrast, we do know many specific hard instances of other problems.
Many such instances come from cryptography; for example, it is hard to invert a one-way function $f$ on a random input,
%(it is widely believed that one-way  functions exist),
it is hard to factor the product $z =xy$ of two large prime numbers $x$ and $y$. Consequently, it is hard to satisfy SAT formulas that encode statements ``$f(x) = y$'' and ``$xy = z$''. There are even more natural families of hard instances of optimization problems; e.g.
\begin{itemize}
\item\textbf{3-SAT:} Feige's 3-SAT Conjecture~\cite{Feige} states that no randomized polynomial time algorithm can distinguish random
instances of 3-SAT (with a certain clause to variable ratio) from $1-\eps$ satisfiable instances of 3-SAT (with non-negligible probability).
\item\textbf{Linear Equations in $\bbZ/2\bbZ$:} Alekhnovich's Conjecture~\cite{Alekhnovich} implies that given a random $(1-\eps)$ satisfiable instance of a system
of linear equations in $\bbZ/2\bbZ$, no randomized polynomial time algorithm can find a solution that satisfies a $1/2 + \delta$ fraction of
equations (for certain values of parameters $\eps$ and $\delta$).
\item\textbf{Maximum Clique Problem:} It is widely believed~\cite{Jerrum} that no randomized polynomial time algorithm can find a clique
of size $(1 + \eps) \log_2 n$ in a $G(n, 1/2)$ graph with a planted clique of size $m = n^{1/2-\delta}$ (for every constant $\eps, \delta > 0$).
\end{itemize}

No such results are known or conjectured for Unique Games. In order to better understand
Unique Games, we need to identify, which instances of the problem are easy and
which are potentially hard. That motivated the study of specific families of Unique Games.
Barak, Hardt, Haviv, Rao, Regev and Steurer~\cite{BHHRRS} showed that unique game instances obtained
by parallel repetition are ``easy'' (we say that a family of $1 - \eps$ satisfiable instances is \textit{easy} if there is a randomized polynomial-time algorithm that satisfies a constant fraction of constraints) . Arora, Khot, Kolla, Steurer, Tulsiani, and Vishnoi~\cite{AKK} showed
that unique games on spectral expanders are easy (see also Makarychev and Makarychev~\cite{MM}, and
Arora, Impagliazzo, Matthews and Steurer~\cite{AIMS}).

In this paper, we investigate the hardness of semi-random (semi-adversarial) instances
of Unique Games.
In a semi-random model, an instance is generated in several steps;
at each step, choices are either made adversarially or randomly.
Semi-random models were introduced by Blum and Spencer~\cite{BS} (who considered semi-random instances
of the $k$-coloring problem) and then studied by Feige and Kilian~\cite{FKil}, and Feige and Krauthgamer~\cite{FKra}.

In this paper, we propose and study a model, in which a $1-\eps$ satisfiable unique game instance is generated as follows:
\begin{enumerate}
\item\textbf{Graph Selection Step.} Choose the constraint graph $G = (V,E)$ with $n$ vertices and $m$ edges.
\item\textbf{Initial Instance Selection Step.} Choose a set of constraints $\{\pi_{uv}\}_{(u,v)\in E}$ so that the obtained instance is completely satisfiable.
\item\textbf{Edge Selection Step.} Choose a set of edges $E_{\eps}$ of size $\eps m = \eps|E|$.
\item\textbf{Edge Corruption Step.} Replace the constraint for every edge in $E_{\eps}$ with a new constraint.
\end{enumerate}
Note that if an adversary performs all four steps, she can obtain an arbitrary $1-\eps$ satisfiable instance, so, in
this fully--adversarial case, the problem is as hard as in the worst case.
The four most challenging semi-random cases are when choices at one out of the four steps are made randomly,
and all other choices are made adversarially. The first case --- when the graph $G$ is random and, in particular, is an expander ---
was studied by Arora, Khot, Kolla, Steurer, Tulsiani, and Vishnoi~\cite{AKK}, who showed that this case is easy.

We present algorithms for the other three cases that with high probability (w.h.p.) satisfy a $1-\delta$ fraction of constraints
(if the average degree of $G$ is at least $\tilde \Omega(\log k)$ and $\eps$ is less than some absolute constant).
\begin{theorem}
For every $k\geq k_0$, $\eps_0 > 0$ and $\delta > C\max(\eps_0,\log k/\sqrt{k})$ (where $C$ and $k_0$ are absolute constants), there exists a randomized polynomial
time algorithm that given a semi-random instance of Unique Games with $\eps = \eps_0$ (generated in one of the
three models; see~Section~\ref{sec:models} for details) on a graph $G$ with average degree at least $\tilde\Omega(\log k) \delta^{-3}$,
finds a solution of value at least $1 - \delta$ with probability%%
\footnote{The probability is over both the random choices that we make when
we generate the semi-random instance, and the random choices that the algorithm does. That is, the probability that the model generates
a semi-random instance $\cal I$, such that the algorithm finds a solution of $\cal I$ of value at least $1-\delta$ with probability $1-o(1)$,
% (over the choice of random bits used by the algorithm),
is $1-o(1)$.}
$1 - o(1)$.
\end{theorem}
The theorem follows from Theorems~\ref{thm:main1}, \ref{thm:main2}, and~\ref{thm:main3}, in which we analyze
each model separately, and establish more precise bounds on the parameters for each model.

In our opinion, this is a very surprising result since the adversary has a lot of control over the semi-random instance. Moreover, our results suggest that the Unique Games problem is different in nature than several NP-hard problems like SAT, which are thought to be hard on average.

We want to point out that previously known approximation algorithms for Unique Games cannot find good solutions of semi-random
instances. Also techniques developed for analyzing semi-random instances of other problems such as local analysis,
statistical analysis, spectral gap methods, standard semidefinite programming techniques seem to be inadequate to deal
with semi-random instances of Unique Games.
To illustrate this point, consider the following example.
Suppose that the set of corrupted edges is chosen at random, and all other steps are adversarial (``the random edges, adversarial constraints case'').
The adversary generates a semi-random instance as follows.
It first prepares two instances ${\cal I}_1$ and ${\cal I}_2$ of Unique Games.
The first instance ${\cal I}_1$ is the Khot--Vishnoi instance~\cite{KV} on a graph $G$ with the label set $[k] = \{0,\dots, k-1\}$ and
permutations $\{\pi^1_{uv}\}$  whose SDP value is $\eps' < \eps/2$ but which is only $k^{-\Omega(\eps')}$
satisfiable. The second instance ${\cal I}_2$ is a completely satisfiable instance on the same graph $G$ with the label set $\{k,\dots, 2k-1\}$
and permutations $\pi_{uv}^2 = \mathrm{id}$. She combines these instances together:
the combined instance is an instance on the graph $G$ with the label set $[2k] = \{0,\dots, 2k-1\}$, and permutations
$\{\pi_{uv}: \pi_{uv}(i) = \pi_{uv}^1(i) \text{ if } i\in[k], \text{ and } \pi_{uv}(i) = \pi_{uv}^2(i), \text{ otherwise}\}$.
Once the adversary is given a random set of edges $E_{\eps}$, she randomly changes (``corrupts'') permutations $\{\pi_{uv}^2\}_{(u,v)\in E_{\eps}}$
but does not change $\pi_{uv}^1$, and then updates permutations $\{\pi_{uv}\}_{(u,v)\in E_{\eps}}$ accordingly.
It turns out that the SDP value of ${\cal I}_2$ with corrupted edges is very close to $\eps$, and therefore,
it is larger than $\eps'$, the SDP value of ${\cal I}_1$ (if we choose parameters properly).
So in this case the SDP solution assigns positive weight only to the labels in $[k]$ from the first instance.
That means that the SDP solution does not reveal any information about the optimal solution
(the only integral solution we can obtain from the SDP solution has value $k^{-\Omega(\eps)}$).
Similarly, algorithms that analyze the spectral gap of the label extended graph cannot deal with this instance.
Of course, in this example, we let our first instance, ${\cal I}_1$, to be the Khot--Vishnoi instance
because it ``cheats'' SDP based algorithms. Similarly, we can take as ${\cal I}_1$ another instance
that cheats another type of algorithms. For instance, if UGC is true, we can let ${\cal I}_1$ to be a $1-\eps'$ satisfiable unique
game that is indistinguishable in polynomial-time from a $\delta$--satisfiable unique game.

Our algorithms work only for values of $\eps$ less than some absolute constants. We show that this
restriction is essential. For every $\eps\geq 1/2$ and $\delta > 0$, we prove that no polynomial time
algorithm satisfies a $\delta$ fraction of constraints in the ``adversarial constraints, random edges'' model (only the third step is random), assuming the 2--to--2 conjecture.

One particularly interesting family of semi-random unique games (captured by our model) are
\textit{mixed instances}. In this model, the adversary prepares a satisfiable instance,
and then chooses a $\delta$ fraction of edges and replaces them with adversarial constraints (corrupted constraints);
i.e. she performs all four steps in our model and can obtain an arbitrary $1-\delta$ satisfiable instance.
Then the ``nature'' replaces every corrupted constraint with the original constraint with probability $1 - \eps$.
In our model, this case corresponds to an adversary who at first prepares a list of corrupted constraints $\pi'_{uv}$, and then
at the fourth step replaces constraints for edges in $E_{\eps}$ with constraints $\pi'_{uv}$
(if an edge from $E_{\eps}$ is not in the list, the adversary does not modify the corresponding constraint).

\textbf{Distinguishing Semi-Random At Most $(1-\eps)$ Satisfiable Instances From Almost Satisfiable Instances.}
We also study whether semi-random instances of Unique Games that are at most $(1-\eps)$ satisfiable can be distinguished
from almost satisfiable instances of Unique Games.
This question was studied for other problems. In particular, Feige's ``Random 3-SAT Conjecture''
states that it is impossible to distinguish between random instances of 3-SAT (with high enough clause density)
and $1-\delta$ satisfiable instances of $3$-SAT.
In contrast, we show that in the ``adversarial edges, random constraints'' case (the fourth step is random),
semi-random $(1-\eps)$-satisfiable instances can be efficiently distinguished from (arbitrary)
$(1-\delta)$-satisfiable instances when $\eps > c\delta$ (for some absolute constant $c$).
(This problem, however, is meaningless in the other two cases --- when the adversary corrupts the constraints ---
since then she can make the instance almost satisfiable even if $\eps$ is large.)

\textbf{Linear Unique Games}
We separately consider the case of Linear Unique Games (MAX $\Gamma$-LIN).
In the semi-random model for Linear Unique Games, we require that constraints chosen at the second and fourth steps are of
the form $x_u - x_v = s_{uv} (\mathrm{mod}\ k)$.
Note that in the ``random edges, adversarial constraints'' model, the condition that constraints are of the
form $x_u - x_v = s_{uv} (\mathrm{mod}\ k)$ only restricts the adversary (and does not change how the random edges are selected).
Therefore, our algorithms for semi-random general instances still applies to this case.
However, in the ``adversarial edges, random constraints'' case, we need to sample constraints from a different
distribution of permutations at the fourth step: for every edge $(u,v)$ we now choose a random shift permutation $x_v = x_u -s_{uv}$,
where $s_{uv} \in_U \bbZ/k\bbZ$. We show that our algorithm still works in this case; the analysis however is different.
We believe that it is of independent interest.
We do not consider the case where only the initial satisfying assignment is chosen at random, since for Linear Unique Games,
the initial assignment uniquely determines the constraints between edges (specifically, $s_{uv} = x_u - x_v (\mathrm{mod}\ k)$).
Thus the case when only the second step is random is completely adversarial.

%Systems of linear equations. Suppose we have a consistent system
%of linear equations $Ax = b$. Then we randomly change an $\eps$ fraction of entries of $b$ (``some equations
%are corrupted by noise''). Our results show that we can efficiently
%find a solution that satisfies a $1-\delta$ fraction of all constraints, if every equation has only two variables.
It is interesting that systems of linear equations affected by noise with more than two variables per equations are believed to be much harder.
Suppose we have a consistent system of linear equations $Ax = b$ over $\bbZ/2\bbZ$. Then we randomly change an $\eps$ fraction of entries of $b$.
Alekhnovich~\cite{Alekhnovich} conjectured that no polynomial-time algorithm can distinguish the obtained instance from a completely random instance
even if $\eps \approx n^{-c}$, for some constant $c$ (Alekhnovich stated his conjecture both for systems with 3 variables per equation and for systems with an arbitrary number of variables per equation).

Our results can be easily generalized to Unique Games in arbitrary Abelian groups. We omit the details
in the conference version of this paper.

\subsection{Brief Overview of Techniques}
In this paper, we develop new powerful algorithmic techniques for solving semi-random instances of unique games.
We use different algorithms for different models. First, we outline how we solve semi-random unique games in
the ``adversarial constraints, random edges'' model (see Section~\ref{sec:adv-constr} for details).
As we explained above, we cannot use the standard SDP relaxation (or other standard techniques) to solve
semi-random instances in this model. Instead, we consider a very unusual SDP program for the problem,
which we call ``Crude SDP'' (C-SDP). This SDP is not even a relaxation for Unique Games
and its value can be large when the instance is satisfiable. The C-SDP assigns a unit vector $u_i$ to every
vertex $(u,i)$ of the label--extended graph (for a description of the label--extended graph we refer the reader
to Section~\ref{sec:prelim}). We use vectors $u_i$ to define the length of edges of the
label--extended graph: the length of $((u,i), (v,j))$ equals $\|u_i - v_j\|^2$.
Then we find \textit{super short edges} w.r.t. the C-SDP solution, those edges that have length $O(1/\log k)$.
One may expect that there are very few short edges since for a given C-SDP most edges
will be long if  we choose the unique games instance at random.
We prove, however, that for every C-SDP solution $\set{u_i}$, with high probability (over the semi-random instance) either
\begin{enumerate}
\item there are many super short edges w.r.t. $\set{u_i}$ in the satisfiable layer of the semi-random game,
\item or there is another C-SDP solution of value less than the value of the solution $\set{u_i}$.
\end{enumerate}
Here, as we describe later on in section \ref{sec:prelim}, the ``satisfiable layer'' corresponds to the representation of the satisfying assignment in the label--extended graph.

Note that if our instance is completely satisfiable, then in the optimal (integral) C-SDP solution all the edges that correspond to the satisfiable layer have length zero and, therefore are super short.

Our proof shows how to combine the C-SDP solution with an integral solution
so that the C-SDP value goes down unless almost all edges in the satisfiable layer are super short.
We then show that this claim holds with high probability not only for one C-SDP solution but also for all
C-SDP solutions simultaneously. The idea behind this step is to find a family $\cal F$ of representative
C-SDP solutions and then use the union bound. One of the challenges is to choose a very small family $\cal F$,
so that we can prove our result under the assumption that the average degree is only $\tilde\Omega(\log k)$.
The result implies that w.h.p. there are many super short edges w.r.t. the optimal C-SDP solution.

Now given the set of super short edges, we need to find which of them lie in the satisfiable layer.
We write and solve an LP relaxation for Unique Games, whose objective function depends
on the set of super short edges. Then we run a rounding algorithm that rounds the C-SDP and LP solutions
to a combinatorial solution using a variant of the ``orthogonal separators'' technique developed in~\cite{CMM2}.

Our algorithm for the ``adversarial edges,  random constraints'' model  is quite different.
First, we solve the standard SDP relaxation for Unique Games.
Now, however, we cannot claim that many
edges of the label--extended graph are short. We instead find the ``long'' edges of the graph $G$ w.r.t. the SDP solution.
We prove that most corrupted edges are long, and there are at most $O(\eps)$ long edges in total (Theorem~\ref{thm:longedges}).
We remove all long edges and obtain a highly--satisfiable instance. Then we write and solve an SDP for this instance, and 
round the SDP solution using the algorithm of~\cite{CMM1}.

We also present algorithms for two more cases: the case of ``adversarial edges, random constraints'' where the constraints are of the special form MAX-$\Gamma$-LIN and the case of ``random initial constraints''.
%%We remark that, due to lack of space the last two cases as well as a fair amount of the %%proofs for the first two cases, are deferred to the Appendix.

In this paper, we develop several new techniques.  In particular, we propose a novel C-SDP program, and then show
how to exploit the optimality of a C-SDP solution in the analysis.  We develop a rounding algorithm that
simultaneously uses SDP and LP solutions. We demonstrate how to bound the number of different SDP solutions
using dimension reduction methods and other tricks. We believe that our techniques are of independent interest and that they
will prove useful for solving semi-random instances of other problems.

%\cstart

%The simultaneous use of an SDP and an LP solution is a novel tool for rounding and we believe that our techniques are of independent interest.\cend

\section{Notation and Preliminaries}
\label{sec:prelim}
\subsection{The Label-Extended Graph}
For a given instance of Unique Games on a constraint graph $G=(V,E)$, with alphabet size $k$ and constraints $\{\pi_{uv}\}_{(u,v)\in E}$ we define the {\em Label-Extended} graph $M(V'=V\times[k],E')$ associated with that instance as follows: $M$ has $k$ vertices $B_v=\{v_0,\cdots{},v_{k-1}\}$ for every vertex $v\in V$. We refer to this set of vertices as the block corresponding to $v$. $M$ has a total of $|V|$ blocks, one for each vertex of $G$. Two vertices $u_i, v_j \in V'$ are connected by an edge if $(u,v)\in E$ and $\pi_{uv}(i)=j$. We refer to a set of nodes $L=\{{u^{(z)}}_{i(z)}\}_{z=1}^{|V|}$ as a ``layer'' if $L$ contains exactly one node from each block $B_{u^{(z)}}$. We note that a layer $L$ can be seen as an assignment of labels to each vertex of $G$. If a layer $L$ consists of vertices with the same index $i$, i.e. $L=\{u^{(z)}_i\}_{z=1}^{|V|}$, we will call $L$ the $i$-th layer.

\subsection{Standard SDP for Unique Games}
Our algorithms use the following standard SDP relaxation for Unique Games (see also~\cite{Kho02, KV, CMM1, CMM2}).
\begin{align*}
&\text{min} &&\frac{1}{2|E|}\sum_{(u,v)\in E}\sum_{i\in [k]} \|u_i - v_{\pi_{uv} (i)}\|^2 &&\\
&\text{subject to}&& \sum_{i=1}^{k} \|u_i\|^2 = 1 && \text{for all } u \in V\\
&&&\langle u_i, u_j\rangle = 0 && \text{for all } u \in V,\ i,j \in [k] ( i\neq j)\\
&&&\langle u_i, v_j\rangle \geq 0 && \text{for all } u,v \in V,\  i,j \in [k] ( i\neq j)\\
&&&\|u_i - v_j\|^2\leq \|u_{i} - w_l\|^2 +\|w_{l} - v_j\|^2 && \text{for all } u,v,w \in V,\   i,j,l\in [k].
\end{align*}
In this relaxation, we have a vector variable $u_i$ for every vertex $u$ and label $i$. In the intended solution,
$u_i$ is an indicator variable for the event ``$x_u =i$''. That is, if $x_u = i$ then $u_i = e$; otherwise, and $u_i = 0$;
where $e$ is a fixed unit vector. The objective function measures the fraction of \textit{unsatisfied} constraints:
if the unique game has value $1 - \eps$, then the value of the intended SDP solution equals $\eps$ (and, therefore,
the value of the optimal SDP solution is at most $\eps$).

Given an SDP solution of value $\eps$, the approximation algorithm of Charikar, Makarychev, and Makarychev~\cite{CMM1}
finds a solution of value $1 - O(\sqrt{\eps \log k})$.
%  and roughly $k^{-\eps/(2-\eps)}$.
We will use this approximation algorithm as a subroutine (we will refer to it as CMMa).
We will also use the following fact.
\begin{lemma}[see e.g. Lemmas A.1 and A.2 in~\cite{CMM2}]\label{lem:gaussian}
Suppose, we are given two random Gaussian variables $\gamma_1$ and $\gamma_2$ with mean $0$ and variance $1$
(not necessarily independent), and a parameter $k\geq 2$. Let $\alpha = 1/(2k^2)$.
Consider a threshold $t$ s.t. $\Pr(\gamma_1 \geq t) = \Pr(\gamma_2 \geq t) = \alpha$. Then
$\Pr(\gamma_1 \geq t \text{ and } \gamma_2 \geq t) \geq \alpha \Bigl(1 - \sqrt{\frac{1}{c^*} \Var(\gamma_1 - \gamma_2)\log k}\Bigr)$
for some absolute constant $c^*$.
\end{lemma}

\subsection{Models}\label{sec:models}
In what follows, we will use several models for generating semi-random $(1-\eps)$ satisfiable instances of Unique Games\footnote{The parameters of the models are the number of vertices $n$, the number of edges $m$, and the probability $\varepsilon$ that an edge is corrupted.}:
\begin{enumerate}
\item \textbf{``Random Edges, Adversarial Constraints'' Model.} The adversary selects a graph $G(V,E)$ on $n$ vertices and $m$ edges and an initial set of constraints $\{\pi_{uv}\}_{(u,v)\in E}$ so that the instance is completely satisfiable. Then she adds every edge of $E$ to a set $E_{\eps}$ with probability $\eps$ (the choices for different edges are independent). Finally, the adversary replaces the constraint for every edge in $E_{\eps}$ with a new constraint of her choice. Note that this model also captures the case where at the last step the constraints for every edge in $E_{\eps}$ are replaced with a new random constraint (random adversary).
\item \textbf{``Adversarial Edges, Random Constraints'' Model.} The adversary selects a graph $G(V,E)$ on $n$ vertices and $m$ edges and an initial set of constraints $\{\pi_{uv}\}_{(u,v)\in E}$ so that the instance is completely satisfiable. Then she chooses a set $E_{\eps}$ of $\eps |E|$ edges. Finally, the constraint for every edge in $E_{\eps}$ is randomly replaced with a new constraint. We will also consider some variations of this model, where at all steps the constraints are MAX $\Gamma$-LIN. In particular, at the last step, choosing a random constraint of the form MAX $\Gamma$-LIN, corresponds to choosing a random value $s\in [|\Gamma|]$.
\item \textbf{``Random Initial Constraints'' Model.} The adversary chooses the constraint graph $G = (V, E)$ and a ``planted solution''
$\set{x_u}$. Then for every edge $(u,v)\in E$, she randomly chooses a permutation (constraint)
$\pi_{uv}$ such that $\pi_{uv}(x_u) = x_v$ (among $(k-1)!$ possible permutations).
Then the adversary chooses an arbitrary set $E_\eps$ of edges of size at most $\eps |E|$ and replaces constraint $\pi_{uv}$ with a constraint $\pi_{uv}'$ of her choice for $(u,v)\in E_\eps$.
\end{enumerate}

\begin{remark}
Without loss of generality, we will assume, when we analyze the algorithms, that the initial completely satisfying assignment corresponds to the ``zero'' layer. I.e. for every edge $(u,v)$, $\pi_{uv}(0)=0$. Note that in reality, the real satisfying assignment is hidden from us.
\end{remark}

\section{Random Edges, Adversarial Constraints}\label{sec:adv-constr}
In this section, we study the ``random edges, adversarial constraints'' model and prove the following result.

\begin{theorem}\label{thm:main1}
Let $k \in \bbN$ ($k\geq 2$), $\varepsilon \in (0,1/3)$,
and $\eta \in (0,1)$.
%and $\gamma \in (\varepsilon, 1)$.
There exists a polynomial-time approximation algorithm, that given an instance of Unique Games from the ``random edges, adversarial constraints'' model on graph $G$  with
$C \eta^{-3}(1/3 - \varepsilon)^{-4} n\log k (\log (\eta^{-1}\log k))^2$ edges ($C$ is a sufficiently
large absolute
constant), finds a solution of value $(1-\varepsilon - \eta)/(1 + \varepsilon + \eta) + O(1/k)$, with probability $1-o(1)$.
\end{theorem}
\begin{corollary}
There exists a polynomial-time approximation algorithm, that given an instance of unique games from the ``random edges, adversarial constraints'' model on graph $G$  with
$Cn\log k (\log\log k)^2$ edges ($C$ is a sufficiently
large absolute
constant) and $\varepsilon \leq 1/4$, finds a solution of value $1/2$, with probability $1-o(1)$.
\end{corollary}
\begin{remark}
We can make the constant in the $O(1/k)$ term in $(1-\varepsilon - \eta)/(1 + \varepsilon + \eta) + O(1/k)$
arbitrarily small by increasing the value of $C$ (and decreasing the value of $\alpha$ in the proof).
We omit the details to simplify the exposition.
\end{remark}

The main challenge in solving ``random edges, adversarial constraints'' unique games is that
the standard SDP relaxation may assign zero vectors to layers corresponding to the optimal
solution (as well as to some other layers) and assign non-zero vectors to layers, where every integral
solution satisfies very few constraints. To address this issue, we introduce a new slightly modified SDP.
As usual the SDP has a vector $u_i$ for every vertex--label pair $(u,i)\in V\times[k]$. We require
that vectors $u_i$, $u_{i'}$ corresponding to the same vertex $u$ are orthogonal: $\langle u_i, u_{i'}\rangle = 0$ for
all $u\in V$ and $i,i'\in[k]$, $i\neq i'$. We also impose triangle inequality constraints:
$$\frac{1}{2}\|u_i - v_j\|^2 + \frac{1}{2}\|u_{i'} - v_j\|^2 \geq 1,$$
for all $(u,v)\in E$ and $i,i',j\in[k]$, $i\neq i'$;
and  require that all vectors have unit
length: $\|u_i\|=1$ for all $u\in V$ and $i\in[k]$. Observe, that our SDP is not a relaxation\footnote{%
Unless, the unique game is from a special family like Linear Unique Games
(see Section~\ref{sec:linGames}).}, since
the integer solution does not satisfy the last constraint.
The objective function is
$$\min \sum_{(u,v)\in E}\sum_{\substack{i\in [k]\\j = \pi_{uv} (i)}}
\frac{\|u_i - v_j\|^2}{2}.$$
Usually, this objective function measures the number of unsatisfied unique games constraints. However, in our
case it does not. In fact, it does not measure any meaningful quantity. Note, that the value of the SDP can be arbitrary large even if the unique games instance is satisfiable.
We call this SDP---the Crude SDP or C-SDP.
Given a C-SDP solution, we define the set of \textit{super short} edges,
which play the central role in our algorithm.
\begin{definition}
We say that an edge $((u,i),(v,j))$ in the label--extended graph is $\eta$--super short, if
$\|u_i-v_j\|^2 \leq c^*\eta^2/\log k$, here $c^*$ is an absolute constant defined in Lemma~\ref{lem:gaussian}. We denote the set of all $\eta$--super short edges by $\Gamma_{\eta}$.
\end{definition}
In Section~\ref{sec:shortedges}, we prove the following surprising result (Theorem~\ref{thm:shortedges}),
which states that all but very few edges in the zero level of the label-extended graph are super short.

\DeclareThmA{shortedges}{
Let $k \in \bbN$ ($k\geq 2$), $c\in(0,1)$, $\varepsilon \in (0,1/3)$,
$\eta \in (c,1)$ and $\gamma \in (\varepsilon + c, 1)$ and let  $G=(V,E)$ be an arbitrary graph
with at least $C \eta^{-2}(\gamma - \varepsilon)^{-1}(1/3 - \varepsilon)^{-4} \times n\log k (\log (c^{-1} \log k))^2$, edges. Consider
a semi-random instance of Unique Games in the ``random edges, adversarial constraints''
model. Let $\{u_i\}$ be the optimal solution of the C-SDP. Then with probability $1-o(1)$, the set
$$\Gamma^0_{\eta} = \Gamma_{\eta} \cap \{((u,0),(v,0)): (u,v)\in E\}$$
contains at least $(1 - \gamma) |E|$ edges.
}

More concretely, we proceed as follows. First, we solve the C-SDP. Then, given the C-SDP solution, we write and solve an LP to obtain weights $x(u,i)\in[0,1]$ for every $(u,i)\in V\times[k]$. These weights are in some sense substitutes for lengths of vectors in the standard SDP relaxation.
In the LP, for every vertex $u\in V$, we require that
$$\sum_{i\in [k]} x(u,i) = 1.$$
The objective function is
$$\max \sum_{((u,i),(v,j))\in \Gamma_{\eta}}\min (x(u,i), x(v,j))$$
(note that the objective function depends on the C-SDP solution).
Denote the value of the LP by $LP$. The intended solution of this LP is $x(u,0)=1$ and
$x(u,i)=0$ for $i\neq 0$. Since the LP contribution of every edge in $\Gamma^0_{\eta}$ is 1,
the value of the intended solution is at least $|\Gamma^0_{\eta}|$.
Applying Theorem~\ref{thm:shortedges} with $\gamma = \eps + \eta$, we get
$|\Gamma^0_{\eta}| \geq (1-\gamma)|E| = (1-\varepsilon -\eta)|E|$, so $LP\geq (1-\varepsilon -\eta)|E|$.
In the next section, we present an approximation algorithm (which rounds C-SDP an LP solutions)
and its analysis. We prove the approximation guarantee in Lemma~\ref{lem:mainAlg1}, which
implies Theorem~\ref{thm:main1}.

\subsection{Lower Bound on the Number of Super Short Edges: Proof of
Theorem~\ref{thm:shortedges}}\label{sec:shortedges}

We need the following lemma.
\newif\iffirstpassonelay
\firstpassonelaytrue
\DeclareLemB{onelay}{
Let $G=(V,E)$ be an arbitrary graph on  $n$ vertices, and let $\varepsilon\in [0,1/3)$, $\rho = 1/3-\varepsilon$, $\nu \in (0,\rho)$. Suppose, that $\{Z_{uv}\}_{(u,v)\in E}$ are i.i.d. Bernoulli random  variables taking values $1$ with probability $\varepsilon$ and $0$  with probability $(1-\varepsilon)$. Define the payoff function
$p: \{0,1\}\times \bbR \to \bbR$ as follows
$$p(z, \alpha) =
\begin{cases}
-2\alpha,&\text{if } z =1;\\
\alpha,&\text{if } z = 0.\\
\end{cases}
$$
Then, with probability at least $1-o(1)$ for every set of vectors $\{u_0\}_{u \in V}$
satisfying (for some significantly large absolute constant $C$)
\begin{equation}\iffirstpassonelay\label{eq:c1}\else\tag{\ref{eq:c1}}\fi
\frac{1}{2}\sum_{(u,v)\in E} \|u_0 - v_0\|^2 \geq C\nu |E| + C \rho^{-4}\log^2 (1/\nu)  n
\end{equation}
the following inequality holds
\begin{equation}\iffirstpassonelay\label{eq:c2}\else\tag{\ref{eq:c2}}\fi
\sum_{(u,v)\in E} p(Z_{uv}, \|u_0 - v_0\|^2) > 0.
\end{equation}
}
\begin{proof}
We need the following dimension reduction lemma, which is based on the Johnson--Lindenstrauss Lemma and
is fairly standard (for an example of using the Johnson--Lindenstrauss Lemma
in SDP rounding see Raghavendra and Steurer~\cite{RS}).

\begin{lemma}\label{lem:JL2}
For every positive $\zeta$, $\eta$ and $\nu$, there exists a set $N$
of unit vectors of size at most
$$\exp\left(O(\zeta^{-2}
\log (1/\eta) \log(1/\nu))\right)$$ such that for every
set of unit vectors $\cal U$ there exists a randomized mapping
$\varphi: {\cal U} \to N$ satisfying the following property:
for every $u,v\in \cal U$,
\begin{equation}\label{eq:JL2}
\Pr ((1+\zeta)^{-1} \|u-v\|^2 - \eta^2 \leq \|\varphi (u) - \varphi (v)\|^2 \leq (1+\zeta)\|u-v\|^2 + \eta^2) \geq 1 - \nu.
\end{equation}
\end{lemma}
\begin{proof}[Proof sketch]
To prove the lemma, we consider an $O(\zeta^{-2} \log (1/\nu))$ dimensional space $L$
and choose a $\eta^2/32$-net $N$ in it. The size of $N$ satisfies the bound in the lemma. To construct the mapping
$\varphi: {\cal U}\to N$, we project all vectors from $\cal U$ to $L$ using the Johnson--Lindenstrauss lemma and then define $\varphi(u)$ to be the closest vector
$u^* \in N$ to the projection of $u$.
\end{proof}

\begin{proof}[Proof of Lemma~\ref{lem:onelay}]
Set $\zeta= \rho/5 \equiv (1/3 - \varepsilon)/5$. Define a new payoff function,
$$p_{\zeta}(z, \alpha) =
\begin{cases}
-2(1+\zeta) \alpha,&\text{if } z =1;\\
(1+\zeta)^{-1}\alpha,&\text{if } z = 0.\\
\end{cases}
$$
Note, that $\zeta < 1/15$ and $p_{\zeta}(0, \alpha) - p_{\zeta}(1, \alpha) = ((1+\zeta)^{-1} - (-2(1+\zeta)))\alpha \leq 4\alpha$ (for $\alpha\geq 0$).
Suppose, that for a given realization $\{Z^*_{uv}\}_{(u,v)\in E}$ of $\{Z_{uv}\}_{(u,v)\in E}$ there exists
a set of unit vectors $\{u_0\}_{u\in V}$
satisfying condition~(\ref{eq:c1}) and violating~(\ref{eq:c2}).
Embed vectors $\{u_0\}_{u\in V}$ into a net $N$ of size $\exp\left(O(\rho^{-2}
\log^2 (1/\nu))\right)$ using Lemma~\ref{lem:JL2},
such that for $(1-\nu^2/2)$ fraction of all edges $(u,v)\in E$,
the following condition holds
$$(1+\zeta)^{-1}\|u_0- v_0\|^2 - \nu^2/8 \leq
 \|u^{*}- v^{*}\|^2 \leq (1+\zeta)\|u_0- v_0\|^2 + \nu^2/4,$$
here $u^*$ is the image of $u_0$; $v^*$ is the image of $v_0$.
Then,
\begin{equation}\label{eq:wtn1}
\sum_{(u,v)\in E} p_{\zeta}(Z^*_{uv}, \|u^{*} - v^{*}\|^2) < \nu^2 |E| < \nu\rho |E| = 5\nu\zeta |E|.
\end{equation}
and (since we choose $C$ to be sufficiently large)
\begin{equation}\label{eq:wtn2}
\sum_{(u,v)\in E} \|u^{*} - v^{*}\|^2 \geq 40\nu |E| + C\rho^{-4}\log^2 (1/\nu) n.
\end{equation}

Thus, the existence of vectors $\{u_0\}_{u\in V}$ satisfying condition~(\ref{eq:c1}) and
violating~(\ref{eq:c2}) implies the existence of vectors $\{u^{*}\}_{u\in V}$ satisfying~(\ref{eq:wtn1})
and~(\ref{eq:wtn2}). We now show that for a random $\{Z_{uv}\}$ such vectors $\{u^{*}\}$
exist with exponentially small probability. Fix a set $\{u^{*}\}_{u \in V} \subset N$ and consider random $\{Z_{uv}\}$,
\begin{eqnarray*}
\E \sum_{(u,v)\in E} p_{\zeta}(Z_{uv}, \|u^{*} - v^{*}\|^2) &\geq&
\sum_{(u,v)\in E} \left((1-\varepsilon)(1-\zeta) - 2 (1+\zeta)\varepsilon\right) \|u^{*} - v^{*}\|^2\\
&=& (1-3\varepsilon -3\varepsilon\zeta -\zeta) \sum_{(u,v)\in E} \|u^{*} - v^{*}\|^2\\
&\geq& (1-3\varepsilon -4\zeta) \sum_{(u,v)\in E} \|u^{*} - v^{*}\|^2\\
&\geq& \zeta \sum_{(u,v)\in E} \|u^{*} - v^{*}\|^2.
\end{eqnarray*}
By Hoeffding's inequality (using that $p_{\zeta}(0, \|u^{*} - v^{*}\|^2) - p_{\zeta}(1, \|u^{*} - v^{*}\|^2) \leq
4 \|u^{*} - v^{*}\|^2$, $\sum_{(u,v)\in E}\|u^{*} - v^{*}\|^2 - 10\nu |E| \geq \frac{3}{4} \sum_{(u,v)\in E} \|u^{*} - v^{*}\|^2$ and $\|u^{*}-v^{*}\|^2 \leq 4$),
\begin{eqnarray*}
\Pr\left(\sum_{(u,v)\in E} p_{\zeta}(Z_{uv}, \|u^{*} - v^{*}\|^2) < 2\nu\rho |E| \right)
&\leq&
\exp\left(-\frac{2\zeta^2\left(\sum_{(u,v)\in E} \|u^{*} - v^{*}\|^2 - 10\nu |E|\right)^2}{\sum_{(u,v)\in E} 16\|u^{*} - v^{*}\|^4}\right)\\
&\leq&
\exp\left(-\frac{\zeta ^2\left(\sum_{(u,v)\in E} \|u^{*} - v^{*}\|^2\right)^2}{\sum_{(u,v)\in E} 64\|u^{*} - v^{*}\|^2}\right)\\
&\leq&
\exp\left(-\frac{\zeta^2}{64}\sum_{(u,v)\in E} \|u^{*} - v^{*}\|^2\right)\\
&\leq&
\exp\left(-\frac{C\zeta^2}{64}\rho^{-4}\log^2 (1/\nu)n\right)\\
&=&
\exp\left(-\frac{C}{25\cdot64}\rho^{-2}\log^2 (1/\nu)n\right)\\
\end{eqnarray*}
The number of all possible subsets $\{u^{*}\}_{u \in V} \subset N$ is at most
$$|N|^n\leq \exp\left(O(n\rho^{-2} \log^2(1/\nu))\right).$$
Hence, by the union bound with probability at least $1-\exp(-n)=1-o(1)$ for random
$\{Z_{uv}\}_{(u,v)\in E}$, there does not exist a set of vectors $\{u_0\}_{u\in V}$
satisfying condition~(\ref{eq:c1}) and violating~(\ref{eq:c2}).
\end{proof}
\end{proof}

\firstpassonelayfalse

\begin{proof}[Proof of Theorem~\ref{thm:shortedges}]
Let $\{u^*_i\}$ be the optimal SDP solution. Pick a unit vector $e$ orthogonal to
all vectors $u^*_i$. Define a new SDP solution $u^{int}_0 = e$ and
$u^{int}_i = u^*_i$ for $i\neq 0$ (for all $u\in V$). Note that restricted to $\{u^{int}_0\}_{u\in V}$ this
solution is integral. Since $\{u^*_i\}$ is the optimal solution,
$$\sum_{(u,v)\in E}\sum_{\substack{i\in [k]\\j=\pi_{uv}(i)}} \|u^*_i - v^*_j\|^2 \leq
\sum_{(u,v)\in E}\sum_{\substack{i\in [k]\\j=\pi_{uv}(i)}} \|u^{int}_i - v^{int}_j\|^2.$$
Denote by $E_{\varepsilon}$ the set of corrupted edges. Let $Z_{uv} = 1$, if $(u,v)\in E_{\varepsilon}$
and $Z_{uv}= 0$, otherwise. Let $\tilde{E}_{\varepsilon} = \{(u,v)\in E: \pi_{uv}(0) \neq 0\}$.
%$\tilde{E}_{\varepsilon} = \{(u,v)\in E: u_0\neq v_0\}$.
Clearly, $\tilde{E}_{\varepsilon}\subset  E_{\varepsilon}$.
Write,
\begin{eqnarray*}
\sum_{(u,v)\in E}\sum_{\substack{i\in [k]\\j=\pi_{uv}(i)}} \|u^*_i - v^*_j\|^2 &-& \|u^{int}_i - v^{int}_j\|^2
= \\
&=&
\sum_{(u,v)\in E\setminus \tilde{E}_{\varepsilon}} \|u^*_0 - v^*_0\|^2 +
\sum_{(u,v)\in \tilde{E}_{\varepsilon}} \|u^*_0 - v^*_{\pi_{uv}(0)}\|^2 + \|u^*_{\pi_{vu}(0)} - v^*_0\|^2
\\&\phantom{=}& -\sum_{(u,v)\in \tilde{E}_{\varepsilon}}
\|u^{int}_0 - v^{int}_{\pi_{uv}(0)}\|^2 + \|u^{int}_{\pi_{vu}(0)} - v^{int}_0\|^2.
\end{eqnarray*}
For $(u,v)\in \tilde{E}_{\varepsilon}$, we have
$\|u^{int}_{0} - v^{int}_{\pi_{uv}(0)}\|^2 = \|u^{int}_{\pi_{vu}(0)} - v^{int}_{0}\|^2= 2$ and
$\|u^*_0 - v^*_{\pi_{uv}(0)}\|^2 \geq 2 - \|u^*_0 - v^*_0\|^2$. Thus,
\begin{eqnarray*}
\sum_{(u,v)\in E}\sum_{\substack{i\in [k]\\j=\pi_{uv}(i)}} \|u^*_i - v^*_j\|^2
- \|u^{int}_i - v^{int}_j\|^2
&\geq&
\sum_{(u,v)\in E\setminus \tilde{E}_{\varepsilon}} \|u^*_0 - v^*_0\|^2 - 2\sum_{(u,v)\in \tilde{E}_{\varepsilon}}
\|u^{*}_0 - v^{*}_0\|^2\\
&\geq&
\sum_{(u,v)\in E\setminus E_{\varepsilon}} \|u^*_0 - v^*_0\|^2 - 2\sum_{(u,v)\in E_{\varepsilon}}
\|u^{*}_0 - v^{*}_0\|^2\\
&=&
\sum_{(u,v)\in E} p(Z_{uv}, \|u^*_0 - v^*_0\|^2),
\end{eqnarray*}
where $p(\cdot,\cdot)$ is the function from Lemma~\ref{lem:onelay}. By Lemma~\ref{lem:onelay},
with probability $1-o(1)$,  for some absolute constant $C'$, $\rho = (1/3-\varepsilon)$,
$\nu = \frac{c^* \eta^2(\gamma- \varepsilon)}{2 C' \log k}$,
$$\frac{1}{2|E|}\sum_{(u,v)\in E} \|u_0^* - v_0^*\|^2 < \frac{C'\nu |E| + C' \rho^{-4}\log^2 (1/\nu)  n}{|E|}
\leq  c^*\eta^2(\gamma- \varepsilon)\log^{-1} k.
$$
The last inequality holds, because
$$|E|\geq C \eta^{-2}(\gamma - \varepsilon)^{-1}\rho^{-4} n\log k (\log(c^{-1} \log k))^2 \geq
\frac{C'}{2c^*}  \eta^{-2}(\gamma - \varepsilon)^{-1}\rho^{-4} n\log k \log^2(1/\nu)^2.$$
By the Markov inequality, for all but $(\gamma - \varepsilon)/2$ fraction of
edges $(u,v)\in E$, $\|u^*_0-v_0^*\|^2 \leq c^*\eta^2/\log k$.
Recall that a pair $((u,0),(v,0))$ is an $\eta$--super short edge
in the label--extended graph if $\pi_{uv}(0) = 0$ (i.e., $(u_0,v_0)$ is an edge of the label--extended graph) and
$\|u^*_0-v_0^*\|^2 \leq c^*\eta^2/\log k$.  By the Chernoff Bound,
$|E_{\varepsilon}|\leq (\varepsilon + (\gamma-\varepsilon)/2)|E|$ with probability $(1-o(1))$; therefore,
$\pi_{uv}(0) \neq 0$ for at most a $\varepsilon + (\gamma-\varepsilon)/2$ fraction of edges.
Thus, with probability $(1-o(1))$, there are at least $(1 - (\gamma-\varepsilon)/2 - (\varepsilon + (\gamma-\varepsilon)/2))m = (1-\gamma) |E|$ $\eta$--super short edges.
\end{proof}

\subsection{SDP and LP Rounding}
We now present an algorithm that given a C-SDP solution $\{u_i\}$ and
an LP solution $\{x(u,i)\}$ finds an integer solution.
We first present a procedure for sampling subsets of vertex--label
pairs, which is an analog of the algorithm for finding orthogonal separators
in CMMb~\cite{CMM2}.

\pagebreak

\rule{0pt}{12pt}
\hrule height 0.8pt
\rule{0pt}{1pt}
\hrule height 0.4pt
\rule{0pt}{6pt}

\noindent \textbf{LP Weighted Orthogonal Separators}

\noindent \textbf{Input:} An SDP solution $\{u_i\}$, an LP solution $\{x(u,i)\}$.

\noindent \textbf{Output:} A set $S$ of label vertex pairs $(u,i)$.

\begin{enumerate}\compactify
\item Set a parameter $\alpha = 1/(2k^2)$, which we call the probability scale.
\item Generate a random Gaussian vector $g$ with independent components distributed as $\calN (0,1)$.
\item Fix a threshold $t$ s.t. $\Pr(\xi \geq t) = \alpha$, where $\xi \sim \calN(0,1)$.
\item Pick a random uniform value $r$ in the interval $(0,1)$.
\item Find set
$$S = \set{(u,i)\in V\times[k]: \langle u_i, g \rangle \geq t \text{ and } x(u,i) \geq r}.$$
\item Return $S$.
\end{enumerate}

\rule{0pt}{1pt}
\hrule height 0.4pt
\rule{0pt}{1pt}
\hrule height 0.8pt
\rule{0pt}{12pt}

The rounding algorithm is given below.

\rule{0pt}{12pt}
\hrule height 0.8pt
\rule{0pt}{1pt}
\hrule height 0.4pt
\rule{0pt}{6pt}

\noindent \textbf{LP and SDP Based Rounding Algorithm}

\noindent \textbf{Input:} An instance of unique games.

\noindent \textbf{Output:} An assignment of labels to the vertices.
\begin{enumerate}\compactify
\item Solve the SDP.
\item Find the set of all super short edges $\Gamma_{\eta}$.
\item Solve the LP.
\item Mark all vertices unprocessed.
\item while (there are unprocessed vertices)
\begin{itemize}
\item Sample a set $S$ of vertex--label pairs using LP weighted orthogonal separators.
\item For all unprocessed vertices $u$:
\begin{itemize}
\item Let $S_u = \{i: (u,i) \in S\}$
\item If $S_u$ contains exactly one element $i$, assign label $i$ to $u$ and mark $u$ as
processed.
\end{itemize}
\end{itemize}
\end{enumerate}
If after $nk/\alpha$ iterations, there are unprocessed vertices, the algorithm assigns arbitrary labels to them and terminates.

\rule{0pt}{1pt}
\hrule height 0.4pt
\rule{0pt}{1pt}
\hrule height 0.8pt
\rule{0pt}{12pt}

\DeclareLemA{ortsep}{
Let $S$ be an LP weighted orthogonal separator. Then, for every $(u,i)\in  V \times [k]$,
\begin{itemize}
\item $\Pr((u,i)\in S) = \alpha x(u,i).$
\end{itemize}
For every $((u,i), (v,j))\in \Gamma_{\eta}$ and $(u,i')\in V\times[k]$ ($i'\neq i$),
\begin{itemize}
\item $\Pr((u,i)\in S \text{ and } (v,j)\in S) \geq
\alpha \min(x(u,i), x(v,j))(1-\eta).$
\item $\Pr((u,i)\in S; (v,j)\in S; (u,i')\in S) \leq \alpha \min(x(u,i), x(v,j))/(2k^2).$
\end{itemize}
}
\begin{proof}
We have
$$\Pr((u,i)\in S) = \Pr(\langle u_i, g \rangle \geq t \text{ and } x(u,i) \geq r)=
\Pr(\langle u_i, g \rangle \geq t)\Pr(x(u,i) \geq r) = \alpha x(u,i).$$
Then, by Lemma~\ref{lem:gaussian},
\begin{eqnarray*}
\Pr((u,i)\in S \text{ and } (v,j)\in S) &=& \Pr(\langle u_i, g \rangle \geq t \text{ and } \langle v_j, g \rangle \geq t)\Pr(\min(x(u,i), x(v,j)) \geq r)\\
&\geq& \alpha (1-\eta) \min(x(u,i), x(v,j)).
\end{eqnarray*}
Finally, we have (below we use that $\langle u_i, g \rangle$ and $\langle u_{i'}, g \rangle$ are independent random
variables)
\begin{eqnarray*}
\Pr((u,i)\in S; (v,j)\in S; (u,i')\in S) &\leq& \Pr(\langle u_i, g \rangle \geq t)
\Pr(\langle u_{i'}, g \rangle \geq t) \Pr(\min(x(u,i),x(v,j)) \geq r)\\
&=& \alpha^2 \min(x(u,i),x(v,j)) \leq \alpha \min(x(u,i), x(v,j))/(2k^2).
\end{eqnarray*}
\end{proof}

\begin{lemma}\label{lem:mainAlg1}
Given a C-SDP solution $\{u_i\}$ and an LP solution $\{x(u,i)\}$ of value at least $LP \geq x |E|$, the algorithm
finds a solution to the unique games instance that satisfies $(1-\eta) x/(2-(1-\eta)x) - O(1/k)$ fraction of
all constraints in the expectation.
\end{lemma}
\begin{proof}
Consider an arbitrary edge $(u,v)$. We estimate the probability that the algorithm
assigns labels that satisfy the constraint $\pi_{uv}$. For simplicity of presentation, assume
that $\pi_{uv}(i) = i$ (we may always assume this by renaming the labels of $v$). Let
$\delta_i(u,v) = \min (x(u,i), x(v,i))$ if $((u,i),(v,i))\in \Gamma_{\eta}$; and $\delta_i(u,v) = 0$, otherwise. Let $\delta(u,v) = \sum_{i} \delta_i(u,v)$ and $\delta'(u,v)
= \delta(u,v) (1-\eta -1/k)$.

Consider an arbitrary
iteration at which both $u$ and $v$ have not yet been processed. By Lemma~\ref{lem:ortsep} (item 2), if
$((u,i), (v,i))\in \Gamma_{\eta}$, then
$$\Pr (i \in S_u \text { and } i \in S_v) \geq \alpha \min (x(u,i), x(v,i)) (1-\eta).$$
Then, by Lemma~\ref{lem:ortsep} (3) and the union bound, the probability that $S_u$ or $S_v$ contains more than
one element and $i \in S_u$, $i \in S_v$ is at most $\alpha \min (x(u,i), x(v,i))/k$. Hence,
the algorithm assigns $i$ to both $u$ and $v$ with probability at least
$$\alpha \min (x(u,i), x(v,i)) (1-\eta - 1/k) = \alpha\delta_i(u,v)\times (1-\eta - 1/k).$$
The probability that the algorithm assigns the same label to $u$ and $v$ is at least
$$\sum_{i: ((u,i), (v,i))\in \Gamma_{\eta}} \alpha\delta_i(u,v)\times (1-\eta - 1/k) =
\alpha \delta'(u,v).$$
The probability that the algorithm assigns a label to $u$ is at most $\alpha$ and similarly
the probability that the algorithm assigns a label to $v$ is at most $\alpha$. Thus the probability
that it assigns a label to either $u$ or $v$ is at most $\alpha(2 -  \delta'(u,v))$.

Hence, the probability that the algorithm assigns the same label to $u$ and $v$ at one of the
iterations is at least
(note that the probability that there are unlabeled vertices when the algorithm stops after $nk/\alpha$ iterations
is exponentially small, therefore, for simplicity we may assume that the number of iterations is infinite)
$$\sum_{t = 0}^{\infty} (1 - \alpha (2-\delta'(u,v)))^t \alpha \delta'(u,v) =
\frac{\alpha \delta'(u,v)}{\alpha (2-\delta'(u,v))} =
\frac{\delta'(u,v)}{2 -\delta'(u,v)}.$$
The function $\delta \mapsto \delta/(2-\delta)$ is convex on $(0,2)$ and
$$\frac{1}{|E|}\sum_{(u,v)\in E} \delta'(u,v) = \frac{1}{|E|}\sum_{(u,v)\in E} \delta(u,v) (1-\eta -1/k) \geq
x (1-\eta -1/k),$$
thus, by Jensen's inequality, the expected number of satisfied constraints is at least
$$\frac{x (1-\eta -1/k)}{2 -x (1-\eta -1/k)} |E|.$$
\end{proof}

\subsection{Hardness: Semi-Random Instances for $\eps \geq 1/2$}
In this section, we show that the problem becomes hard when $\eps \geq 1/2$ assuming the 2--to--2 conjecture.
The 2--to--2 conjecture follows from Khot's 2--to--1 conjecture\footnote{
The 2--to--1 conjecture implies the 2--to--2 conjecture because every 2--to--1 game can be converted to
a 2--to--2 game of the same value as follows. Let $\calI$ be an instance of a 2--to--1 game on a graph $(L, R, E)$,
s.t. every constraint has degree 2 on the left side, and degree 1 on the right side.
We create two copies $a_1$ and $a_2$ for every label $a$ for vertices in $L$. We replace each constraint $\Pi_{lr}$ with a constraint
$\Pi_{lr}^*$ defined by $\Pi_{lr}^*(a_1, b) = \Pi_{lr}^*(a_2, b) = \Pi_{lr}(a, b)$. We obtain a 2--to--2 game. It is clear that
its value equals the value of the original 2--to--1 game.},
which is a frequently used complexity assumption~\cite{Kho02}.
We prove the following theorem.
\begin{theorem}\label{thm:hardness}
For every $\eps \geq 1/2$ and $\delta >0$, no polynomial-time algorithm can distinguish with probability greater than $o(1)$ between the following two cases:
\begin{enumerate}
\item Yes Case: the instance is a $1-\eps$ satisfiable semi-random instance (in the ``random edges, adversarial constraints'' model),
\item No Case: the instance is at most $\delta$ satisfiable.
\end{enumerate}
This result holds if the 2--to--2 conjecture holds.
\end{theorem}

Before we proceed with the proof we remind the reader the definition of 2--to--2 games and 2--to--2 conjecture.
\begin{definition} In a 2--to--2 game, we are given a graph $G=(V,E)$, a set of labels $[k]=\{0,\dots, k-1\}$ ($k$ is even)
and set of constraints, one constraint for every edge $(u,v)$. Each constraint is defined by a 2--to--2 predicate
$\Pi_{uv}$: for every label $i$ there are exactly two labels $j$ such that $\Pi_{uv}(i,j) = 1$ (the predicate is satisfied);
similarly, for every $j$ there are exactly two labels $i$ such that $\Pi_{uv}(i,j) = 1$.
Our goal is to assign a label $x_u\in [k]$ to every vertex $u$ so as to maximize the number of satisfied constraints $\Pi_{uv}(x_u, x_v) = 1$.
The value of the solution is the number of satisfied constraints.
\end{definition}
\begin{definition} The 2--to--2 conjecture states that for every $\delta > 0$ and sufficiently large $k$, there is no polynomial time
algorithm that distinguishes between the following two cases (i) the instance is completely satisfiable and (ii) the instance
is at most $\delta$ satisfiable.
\end{definition}

\begin{proof}
We construct a randomized reduction that maps every satisfiable 2--to--2 game to a semi-random unique game,
and every at most $\delta$ satisfiable 2--to--2 game to an at most $\delta$ satisfiable unique game.
Suppose we are given an instance of a 2--to--2 game on a graph $G$ with predicates $\Pi_{uv}$. For each
predicate $\Pi_{uv}$, we find two permutations $\pi^0_{uv}$ and $\pi^1_{uv}$ such that
$\Pi_{uv}(i, \pi^0_{uv}(i)) = 1$ and $\Pi_{uv}(i, \pi^1_{uv}(i)) = 1$, and $\pi^0_{uv}(i) \neq \pi^1_{uv}(i)$
for every label $i$ as follows. We consider a bipartite graph on vertex set $\{(u,i): i\in[k]\} \cup \{(v,j): j\in [k]\}$,
in which two vertices $(u,i)$ and $(v,j)$ are connected with an edge if $\Pi_{uv}(i, j) = 1$. Every vertex has degree 2 in this graph.
Therefore, it is a union of two matchings; each of them defines a permutation.
Now for every edge $(u,v)$, we choose one of the two permutations, $\pi^0_{uv}$ or $\pi^1_{uv}$ , at random.
We obtain a unique game instance.

First of all it is clear, that the value of the unique game is at most the value of the 2--to--2 game
(if a labeling $\{x_u\}$ satisfies $x_v = \pi^0_{uv}(x_u)$ or $x_v = \pi^1_{uv}(x_u)$ then it also satisfies
$\Pi_{uv}(x_u, x_v)$). Hence our reduction maps an at most $\delta$ satisfiable 2--to--2 game
to an at most $\delta$ satisfiable unique game.
Now suppose that the 2--to--2 game instance is completely satisfiable. We show how an adversary can
generate semi-random instances that have the same distribution as instances generated by our reduction.
The adversary finds a solution $\{x_u^*\}$ to the 2--to--2 game that satisfies all constraints $\Pi_{uv}$. For every edge $(u,v)$,
she chooses $r_{uv}\in \set{0,1}$ such that $\pi^{r_{uv}}_{uv}(x_u^*) = x_v^*$.
She obtains a completely satisfiable unique game on $G$ with constraints
$\pi^{r_{uv}}_{uv}(x_u) = x_v$. Now she performs the random step --- she
chooses a random set of edges $E_{\eps}$; every edge belongs to $E_{\eps}$ with probability $\eps$.
She replaces every constraint $\pi^{r_{uv}}_{uv}(x_u) = x_v$ with the constraint
$\pi^{1 - r_{uv}}_{uv}(x_u) = x_v$ with probability $1/(2\eps)$, and returns
the obtained instance. Note that the constraint for the edge $(u,v)$ is $\pi^{0}_{uv}(x_u) = x_v$ with probability $1/2$
and $\pi^{1}_{uv}(x_u) = x_v$ with probability $1/2$. Therefore, the distribution of instances coincides
with the distribution generated by our reduction.
\end{proof}

Our positive results (for $\eps < 1/3$) apply only to graphs with average degree greater than
$\eps^{-1} \log k\log \log k$. We want to point out that our negative results (for $\eps \geq 1/2$)
also apply to graphs with high average degree. Indeed, note that our reduction does not change the constraint
graph. As the following lemma shows, if the 2--to--2 conjecture is true then it is also true for instances on graphs with very high
average degree (as large as $n^{1-\theta} \gg k\log k$).

\begin{lemma}
Suppose that the 2--to--2 conjecture holds. Then it also holds for graphs
with average degree at least $n^{1 - \theta}$ (for every $\theta > 0$).
\end{lemma}
\begin{proof}
Suppose we are given an instance $\cal I$ of a 2--to--2 game on a graph $G = (V,E)$ with constraints $\Pi_{uv}(x_u, x_v)$.
We construct a new instance $\cal I'$ of a 2--to--2 game as follows.
For every vertex $u$, we create a ``cloud'' of $N = |V|^{2/\theta}$ new vertices $V_u$.
We connect all vertices in $V_u$ with all vertices in $V_v$ for every edge $(u,v)$.
For every edge between vertices $a\in V_u$ and $b\in V_v$, we add the constraint
$\Pi_{uv}(x_a, x_b) =1$ (where $\Pi_{uv}$ is the predicate for the edge $(u,v)$ in $\cal I$).
We obtain a 2--to--2 game on a graph with average degree
$2N^2 |E| /(N |V|) \geq 2N / |V| \geq (N|V|)^{1 -\theta}$.

Clearly, this is a polynomial--time reduction. We show that the values of $\calI$ and $\calI'$ are equal.
If a solution $\{x_u\}$ for $\cal I$ has value $t$ then the solution defined by $x_a = x_u$ if $a\in V_u$ for $\cal I'$
has also value $t$. On the other hand, given a solution $\{x_a\}$ for $\cal I'$ of value $t$, we construct a solution
for $\cal I$ as follows: for every $u$, we choose a random vertex $a$ from $V_u$ and let $x_u = x_a$.
Then the expected value of this solution is $t$. Therefore, the value of the optimal solution for $\calI$ is at least $t$.
\end{proof}

\section{Adversarial Edges, Random Constraints}

\begin{theorem}\label{thm:main2}
There exists a polynomial-time approximation algorithm that given
$k \in \bbN$ ($k\geq k_0$), $\varepsilon \in (0,1)$, $\eta \in (c \sqrt{\varepsilon\log k /k}, 1)$
and a semi-random instance of unique games from the ``adversarial edges,  random constraints''
model on
graph $G=(V,E)$ with at least $C\eta^{-2} n\log k$ edges ($C$ is
a sufficiently large absolute constant)
finds a solution of value $1 - O(\eps + \eta)$, with probability $1-o(1)$.
\end{theorem}

Our algorithm proceeds in several steps. First, it solves the standard SDP relaxation for Unique Games.
Then it removes ``$\eta$--long edges'' (see below) with respect to the SDP solution, and finally
it runs the CMMa~\cite{CMM1} algorithm to solve the unique game on the remaining graph (the CMMa
algorithm will again solve the SDP relaxation for Unique Games --- it cannot reuse our SDP solution).

\begin{definition}
We say that an edge $(u,v)$ is $\eta$--long with respect to an SDP solution $\{u_i\}$,
if
$$\frac{1}{2}\sum_{i\in [k]} \|u_i - v_{\pi_{uv}(i)}\|^2 > \eta.$$
Otherwise, we say that the edge $(u,v)$ is $\eta$--short.
\end{definition}
Now we formally present the algorithm.

\rule{0pt}{12pt}
\hrule height 0.8pt
\rule{0pt}{1pt}
\hrule height 0.4pt
\rule{0pt}{6pt}

\noindent \textbf{Input:} An instance of unique games.

\noindent \textbf{Output:} An assignment of labels to the vertices.
\begin{enumerate}\compactify
\item Solve the SDP and obtain an SDP solution $\{u_i^*\}$.
\item Remove all $1/16$--long (with respect to $\{u_i^*\}$) edges $(u,v)\in E$ from the graph $G$. Denote the new graph
$G^*$.
\item Solve the SDP on the graph $G^*$ and run the CMMa algorithm.
\end{enumerate}

\rule{0pt}{1pt}
\hrule height 0.4pt
\rule{0pt}{1pt}
\hrule height 0.8pt
\rule{0pt}{12pt}

In Theorem~\ref{thm:longedges}, we show that after removing all $1/16$--long edges
from the graph $G$, the unique games instance contains at most $O(\gamma)|E|$ corrupted
constraint w.h.p, where $\gamma = \eta^2/\log k$.
Since the value of the optimal SDP is at least $\varepsilon$, the algorithm removes at most
$16\varepsilon$ edges at step 2. In the remaining graph, $G'$, the CMMa algorithm finds
an assignment satisfying $1 - O(\sqrt{\gamma \log k}) = 1- O(\eta)$ fraction of all constraints. This assignment satisfies at least $1-O(\varepsilon + \eta)$ fraction of all constraints in $G$.

\begin{remark}
In the previous section we proved that a typical instance of unique games
in the ``random edges, adversarial constraints'' model contains many ``super short'' edges of the label-extended
graph. Then we showed how we can find an assignment satisfying many super short edges. Note, that edges in the set $E_{\varepsilon}$ are not necessarily short or long. In this section, we show something very different: in the typical instance of unique games
in the ``adversarial edges,  random constraints'' model, most edges in the set
$E_{\varepsilon}$ are long. However, note that the label-extended graph does not have to have any super short edges at all.
\end{remark}

\subsection{Almost All Corrupted Edges are Long}

\begin{theorem}\label{thm:longedges}
Let $k \in \bbN$ ($k\geq k_0$), $\varepsilon \in (0,1]$, $\gamma \in (c\varepsilon/\sqrt{k}, 1/\log k)$. Consider
a graph $G=(V,E)$ with at least $C \gamma^{-1} n$ edges
and a unique game instance on $G$ ($c$, $C$, and $k_0$ are absolute constants).
Suppose that all constraints for edges in $E_{\eps}$ are chosen at random;
where $E_{\eps}$ is a set of edges of size $\eps|E|$. Then, the set $E_{\varepsilon}$ contains
less than $\gamma |E|$ $1/16$--short edges w.r.t.  every SDP solution $\{u_i\}$
with probability $1-o(1)$.
\end{theorem}
\begin{proof}
Consider a semi-random instance of unique games. Let $\{u_i\}$ be an SDP solution.
Suppose, that at least $\gamma|E|$ edges in $E_{\varepsilon}$ is $1/16$--short.
Define,
$$f_{\alpha,\{u_i\}}(u,i,v) =
\begin{cases}
1,& \text{if } \|u_i- v_{\pi_{uv}(i)}\|^2 < \alpha (\|u_i\|^2 + \|v_{\pi_{uv}(i)}\|^2);\\
0,&\text{otherwise}.
\end{cases}
$$
We shall omit the second index of $f$, when it is clear that we are talking about the set
$\{u_i\}$. For every short $1/16$--short edge $(u,v)\in E_{\varepsilon}$, we have
\begin{equation}\label{eq:f-sum1}
\frac{1}{2}\sum_{\substack{i\in [k]\\j = \pi_{uv} (i)}} \frac{\|u_i\|^2 + \|v_j\|^2}{8} \times(1- f_{1/8} (u,i,v))\leq
\frac{1}{2}\sum_{\substack{i\in [k]\\j = \pi_{uv} (i)}} \|u_i - v_j\|^2 \leq \frac{1}{16},
\end{equation}
and, hence,
\begin{equation}\label{eq:f-sum2}
\frac{1}{2}\sum_{\substack{i\in [k]\\j = \pi_{uv} (i)}} (\|u_i\|^2 + \|v_j\|^2) f_{1/8} (u,i,v)
\geq
\frac{1}{2}\sum_{\substack{i\in [k]\\j = \pi_{uv} (i)}} (\|u_i\|^2 + \|v_j\|^2) - \frac{1}{2} = \frac{1}{2}.
\end{equation}
We get (from (\ref{eq:f-sum2}) and the
 assumption that there are at least $\gamma |E|$  $1/16$--short edges in $E_{\varepsilon}$)
\begin{eqnarray*}
\sum_{(u,i) \in V\times [k]}\sum_{v:(u,v)\in E_{\varepsilon}}
\|u_i\|^2 f_{1/8}(u,i,v)&=&
\sum_{(u,v)\in E_{\varepsilon}}\sum_{\substack{i\in [k]\\j = \pi_{uv} (i)}} (\|u_i\|^2 + \|v_j\|^2) f_{1/8} (u,i,v)\geq \gamma |E|.\\
&=& \frac{\gamma}{2} |E| +
\frac{\gamma}{8} \sum_{(u,i) \in V\times [k]} \|u_i\|^2
(\varepsilon^{-1}\deg_{E_{\varepsilon}}(u)+ \Delta).
\end{eqnarray*}
Here $\deg_{E_{\varepsilon}}(u)$ denotes the number of edges in $E_{\varepsilon}$ incident to $u$; and $\Delta$ denotes the average degree in the graph $G$. We used that $\sum_{i\in[k]}\|u_i\|^2 = 1$.
We now group all vertices of the label--extended graph $(u,i)\in E\times[k]$  depending on the value of $\log_2 (\|u_i\|^2)$. Specifically, we pick a random $r\in (1,2)$ distributed with density $(x \ln 2)^{-1}$
and for
every $t\in \bbN$, define
$$V_{t,r} = \{(u,i)\in V\times[k]: \|u_i\|^2 \in [r2^{-(t+1)}, r 2^{-t}]\}.$$
We define $V_{\infty,r} = \{(u,i)\in V\times[k]: \|u_i\|^2=0\}$.
Let $g_r(u,i,v) = 1$, if $(u,i)$ and $(v, \pi_{uv}(i))$ lie in the same group $V_{t,r}$; and
$g_r(u,i,v) = 0$, otherwise. We claim, that if
$f_{1/8}(u,i,v) = 1$, then $\Pr_r(g_r(u,i,v) = 1) \geq 1/2$.

\begin{claim}
If $f_{1/8}(u,i,v) = 1$, then $\Pr_r(g_r(u,i,v) = 1) \geq 1/2$.
\end{claim}
\begin{proof}
Denote $j= \pi_{uv}(i)$. If $f_{1/8}(u,i,v) = 1$, then $u_i\neq 0$ and $v_j\neq 0$. Assume
w.l.o.g. that $\|u_i\|^2\geq \|v_j\|^2$. Note, that $g_r(u,i,v)=0$, if and only if
$r2^{-t}\in [\|v_j\|^2,\|u_i\|^2]$ for some $t\in \bbN$. The probability of
this event is bounded by $\ln(\|u_i\|^2/\|v_j\|^2)/\ln 2$, since the probability
that an interval $[x,x+\Delta x]$ contains a point from $\{r2^{-t} :t\in \bbN\}$
is $\Delta x/(x\ln(2)) + o(1)$. Using the inequality
$$\|u_i\|^2 - \|v_j\|^2 \leq
 \|u_i - v_j\|^2 \leq \frac{\|u_i\|^2 + \|v_j\|^2}{8},$$
we get $\|u_i\|^2/\|v_j\|^2 \leq 9/7$. Thus,
$$\Pr_r(g_r(u,i,v) = 0) \leq \ln(9/7)/\ln 2 < 1/2.$$
\end{proof}

\noindent Therefore, for some $r\in[1,2]$,
$$\sum_{(u,i) \in V\times[k]}\sum_{v:(u,v)\in E_{\varepsilon}}
\|u_i\|^2 f_{1/8}(u,i,v)g_r(u,i,v)
\geq \frac{\gamma}{4} |E| +
\frac{\gamma}{16}
\sum_{(u,i) \in V\times[k]} \|u_i\|^2(\varepsilon^{-1}\deg_{E_{\varepsilon}}(u)+  \Delta).$$
Rewrite this inequality as follows
\begin{multline*}
\sum_{t\in \bbN}\sum_{(u,i) \in V_{t,r}\times[k]}\sum_{v:(u,v)\in E_{\varepsilon}}
\|u_i\|^2 f_{1/8}(u,i,v)g_r(u,i,v)
\geq\\ \sum_{t\in \bbN}2^{-(t+1)} \frac{\gamma}{4} |E| +
\sum_{t\in \bbN}
\sum_{(u,i) \in V_{t,r}} \|u_i\|^2(\varepsilon^{-1}\deg_{E_{\varepsilon}}(u)+  \Delta).
\end{multline*}
For some $t\in \bbN$,
$$\sum_{(u,i) \in V_{t,r}}\sum_{v:(u,v)\in E_{\varepsilon}}
\|u_i\|^2 f_{1/8}(u,i,v)g_r(u,i,v)
\geq \frac{2^{-(t+1)}\gamma|E|}{4} +
\frac{\gamma}{16}
\sum_{(u,i) \in V_{t,r}}\|u_i\|^2(\varepsilon^{-1}\deg_{E_{\varepsilon}}(u)+  \Delta).$$
We now replace each term $\|u_i\|^2$  in the left hand side with the upper bound $r 2^{-t}$ and each  term
$\|u_i\|^2$  in the right hand side with the lower bound $r 2^{-(t+1)}$. Then, we divide both sides
by $r 2^{-t}\leq 2^{-t}$.
$$\sum_{(u,i) \in V_{t,r}}\sum_{v:(u,v)\in E_{\varepsilon}}
f_{1/8}(u,i,v)g_r(u,i,v)
\geq \frac{\gamma}{16}|E| + \frac{\gamma}{32} \sum_{(u,i) \in V_{t,r}}\varepsilon^{-1}\deg_{E_{\varepsilon}}(u)+ \frac{\gamma}{32} |V_{t,r}|\Delta.$$
Define a new set of vectors $u^*_i = u_i/\|u_i\|$, if $(u,i)\in V_{t,r}$; $u^*_i = 0$, otherwise.
Observe, that
if $f_{\alpha}(u,i,v)g_r(u,i,v) = 1$, then $f_{\alpha,\{u^*_i\}}(u,i,v) = 1$ since for $j=\pi_{uv}(i)$,
\begin{eqnarray*}
\|u^*_i - v^*_j\|^2 &=& 2 - 2 \langle u^*_i, v^*_j \rangle =
2 - 2 \frac{\langle u_i, v_j\rangle}{\|u_i\|\;\|v_j\|}  = 2 - \frac{\|u_i\|^2 + \|v_j\|^2 - \|u_i - v_j\|^2}{\|u_i\|\;\|v_j\|}\\
&< & 2 - (1-\alpha) \frac{\|u_i\|^2 + \|v_j\|^2}{\|u_i\|\;\|v_j\|}\leq 2 - (1-\alpha) \cdot 2 = 2\alpha
= \alpha (\|u^*_i\|^2 + \|v^*_j\|^2).
\end{eqnarray*}
Therefore,
\begin{equation}
\sum_{(u,i)\in V_{t,r}}\sum_{v:(u,v)\in E_{\varepsilon}}f_{1/8,\{u_i^*\}}(u,i,v)
\geq \frac{\gamma}{16}|E| + \frac{\gamma}{32}
\sum_{(u,i) \in V_{t,r}} \varepsilon^{-1}\deg_{E_{\varepsilon}}(u) + \frac{\gamma}{32}|V_{t,r}|\Delta.
\end{equation}
We now embed vectors $\{u^*_i\}$ in a net $N$ of size $\exp(O(\log k))$
using a randomized mapping $\varphi$ (see Lemma~\ref{lem:JL2}, a variant of the Johnson--Lindenstrauss lemma), so
that for some small absolute constant $\beta$ and every $(u,i),(v,j)\in V_{t,r}$,
$$
\Pr ((1+\beta)^{-1} \|u-v\|^2 - \beta \leq \|\varphi (u) - \varphi (v)\|^2 \leq (1+\beta)\|u-v\|^2 + \beta) \geq 1 - \frac{\beta}{k^2}.
$$
We say that a vertex $u\in V$ is good if for all $i',i''\in \{i: (u,i)\in V_{t,r}\}$,
the following inequality holds: $\|\varphi(u_{i'}^*)- \varphi(u_{i''}^*) \|^2 \geq 2 - 3\beta$ for $i'\neq i''$.
I.e., if $u$ is good then all vectors $\varphi({u_i}^*)$ are almost orthogonal. By Lemma~\ref{lem:JL2},
vertex $u$ is good with probability at least $1-\beta$. We let $u^{**}_{i} = \varphi(u^{*}_{i})$, if $u$ is good; and  $u^{**}_{i} = 0$, otherwise.  To slightly simplify the proof,
for every vertex $u$, we also zero out a random subset of $\roundup{\beta k/4}$ vectors $u^{**}_i$. Thus,
if $f_{1/8, \set{u_i^*}}(u,i,v)=1$, then $\Pr(f_{1/4, \set{u_i^{**}}}(u,i,v) = 1)\geq 1 - 4\beta$
(unless $u$ or $v$ are not good vertices; $u_i$ or $v_{\pi_{uv}(i)}$ has been zeroed; or the distance between $u_i^*$ and $v_{\pi_{uv}(i)}^*$
is significantly distorted by $\phi$).
Hence, for some sample $\{u^{**}_{i}\}$,
\begin{equation}\label{eq:goodform}
R \equiv \sum_{(u,i): \|u_i^{**}\|=1}\sum_{v:(u,v)\in E_{\varepsilon}}
f_{1/4,\{u_i^{**}\}}(u,i,v)
\geq \frac{\gamma}{64}\Big(
|E| + \sum_{(u,i) \in V_{t,r}} \varepsilon^{-1}\deg_{E_{\varepsilon}}(u)
+ |V_{t,r}| \Delta \Big)\equiv S.
\end{equation}
We denote the left hand side by $R$. We denote the right hand side by $S$.

We now fix the set of vectors $\{u_i^{**}\}$ and estimate the probability that for a given
set $\{u_i^{**}\}$ and a random set of constraints $\{\pi_{uv}\}$ on edges $(u,v)\in E_{\varepsilon}$
inequality~(\ref{eq:goodform}) holds.
%We set $\gamma=c^*\eta/\log k$. %--we already defined $\gamma$
For each $(u,v)\in E_{\varepsilon}$ and $i$,
$\Pr_{\set{\pi_{uv}}}(f_{1/4,\{u_i^{**}\}}(u,i,v) = 1) \leq 1/k$ (since for every $i$, there is at most one
$j\in[k]$ such that $\|u^{**}_i-v^{**}_j\|^2 \leq 1/4$; here we use that all non-zero vectors $v_j$ are
almost orthogonal). Thus,
$$
\E R = \E\sum_{(u,i): \|u_i^{**}\|^2=1}\sum_{v:(u,v)\in E_{\varepsilon}}f_{1/4,\set{u_i^{**}}}(u,i,v)
\leq \frac{1}{k}\sum_{(u,i): \|u_i^{**}\|^2=1} \deg_{E_{\varepsilon}}(u) < 64 S/(\sqrt{k}c).
$$
In the last inequality, we used the bound $\gamma \geq c\varepsilon/\sqrt{k}$ and (\ref{eq:goodform}).
We assume that $64/c\ll 1$.
We would like to apply the Bernstein inequality, which would give us an upper bound
of
\begin{equation}\label{ineq:want}
\exp(-\Omega(S)) \equiv \exp\Big(-\Omega\Big(\big(\gamma |E| + \gamma \sum_{(u,i) \in V_{t,r}} \varepsilon^{-1}\deg_{E_{\varepsilon}}(u) + \gamma |V_{t,r}| \Delta\big)\log (k) \Big)
\end{equation}
on the probability that inequality~(\ref{eq:goodform}) holds (note: $\E[R] \ll S$;
$\log (S/\E[R])\geq \Theta(\log k)$). Formally,
we cannot use this inequality, since
random variables $f_{1/4, \set{u_i^{**}}}(u,i,v)$ are not independent. So, instead,
we consider a random process, where we consequently pick edges
$(u,v)\in E_{\varepsilon}$ and then for every $i$ such that $u^{**}_i \neq 0$
pick a random yet unchosen $j\in[k]$ and set $\pi_{uv}(i) = j$. For every $u$ and $i$,
there are at least $\roundup{\beta k/4}$ different candidates $j$ (since at least $\roundup{\beta k/4}$
$u_i$'s equal $0$) and for at most one of them
$\|u^{**}_i-v^{**}_j\|^2 \leq 1/4$. Hence, at every step
$\Pr(f_{1/4, \set{u_i^{**}}}(u,i,v)=1)\leq 4/(\beta k)$.
We now apply a martingale concentration inequality and get the
same bound~(\ref{ineq:want}).

Finally, we bound the number of possible solutions $\{u_i^{**}\}$ and then
apply the union bound to show that w.h.p. the set $E_{\varepsilon}$ contains less than $\gamma |E|$ 1/16-short
edges with respect to every SDP solution.
Let $W=|V_{t,r}|$ be the number of pairs $(u, i)$ such that $u_i^{**} \neq 0$.
We choose $W$ variables among $\{u_i^{**}\}_{u,i}$ and assign them values from $N$ as follows.
Let $a_u= |\{u_i^{**}\neq 0: i\}|$. First, we choose the values of $a_u$ for all vertices
(note that $\sum_{u\in V} a_u = W$).
The number of ways to choose $\{a_u\}$ is at most the number of ways to write $W$ as the sum of
$n$ numbers, which is $\binom{n+W-1}{n-1} < 2^{n + W}$.
Then for each vertex $u$, we choose $a_u$ labels ${i_1}, \dots, {i_{a_u}}$ for which
$u_{i_1}^{**} \neq 0, \dots, u_{i_{a_u}}^{**} \neq 0$. The number of ways to choose labels is at most $\prod_{u\in V} k^{a_u} = k^{W}$.
Finally, we assign a vector from $N$ to each chosen variable $u_i^{**}$.  The number of ways
to do so is $|N|^W$. Thus there are at most
$$2^{n+W} \times |N|^W\times k^W = \exp(O(W \log k + n))$$
ways to choose vectors $\{u^{**}_{i}\}$ for given $W$.
Since $\gamma |V_{t,r}| \Delta\log k \geq \gamma W \times (C\gamma^{-1})
\log k\geq CW\log k$ and $\gamma |E|\log k \geq Cn$ inequality~(\ref{eq:goodform})
holds for some collection $\{u^{**}\}$ with exponentially small probability $\exp(-\Omega(n))$ for every fixed value
of $W$. Since there are only $n$ possible values on $W$, inequality~(\ref{eq:goodform}) holds for some value of $W$
with exponentially small probability.
\end{proof}

\section{Random Corrupted Constraints in Linear Unique Games}\label{sec:linGames}

In this section we study a special family of unique games, so-called Linear Unique Games (or
MAX $\Gamma$-Lin). In Linear Unique Games all labels $i\in [k]$ are elements of the cyclic
group $\bbZ/k\bbZ$; all constraints have the form $x_u = x_v - s_{uv}$,
where $s_{uv}\in \bbZ/k\bbZ$. It is known that if the Unique Games Conjecture
holds for general unique games, then it also holds for this family. We note that
the results from the previous subsection are not directly applicable for
Linear Unique Games, since random permutations are restricted to be shifts
$\pi_{uv}: x_u \mapsto x_v + s_{uv}$ and, thus, a random unique game
from the ``adversarial edges, random constraints'' model is most likely not
a linear unique game. This is not the case in the ``random edges, adversarial constraints'' model,
since there the adversary may choose
corrupted constraints in any fashion she wants, particularly to be shifts.
We note that the results of this section can be easily generalized to any
finite Abelian group using the decomposition of the group into a sum of
cyclic subgroups. We omit the details from the conference version of the paper.

For Linear Unique Games, we use a shift invariant variant of the SDP relaxation (see e.g., a paper by Andersson, Engebretsen, and H{\aa}stad~\cite{AEH}).
In this relaxation, all vectors $u_i$ are unit vectors, satisfying $\ell_2^2$ triangle
inequalities, and an extra constraint $\|u_i-v_j\| = \|u_{i+s}-v_{j+s}\|$. The objective is to minimize
$$ \frac{1}{k} \times \frac{1}{2|E|} \sum_{(u,v)\in E}\sum_{i\in [k]}\|u_i-v_{i+s_{uv}}\|^2.$$
Given an arbitrary SDP solution $\{u_i\}$ one can obtain a uniform
solution using symmetrization (this trick works only for unique games on groups, not
arbitrary unique games):
$$u'_i = \frac{1}{\sqrt{k}}\bigoplus_{j\in \bbZ/k\bbZ} u_{i + j}.$$
It is easy to verify (and it is well-known, see e.g.,~\cite{CMM1}) that vectors $\{u'_i\}$ satisfy
all SDP constraints and the objective value of the SDP does not change. We adapt the definition of $\eta$--short edges for this variant of SDP as follows: we say that an edge $(u,v)$
is $\eta$--short with respect to the SDP solution $\{u_i\}$ if $1/2\;\|u_i - v_{i+s}\|^2\leq \eta$ for every $i\in [k]$ (or, equivalently, for some $i\in [k]$, since $\|u_i - v_{i+s}\| = \|u_j - v_{j+s}\|$).

We prove the following analog of Theorem~\ref{thm:longedges}, which immediately implies
that our general algorithm for the adversarial edges, random constraints model works
for linear unique games.

\begin{theorem}\label{thm:longedgesLin}
Let $k \in \bbN$ ($k\geq k_0$), $\varepsilon \in (0,1)$, $\eta \in (\varepsilon/\log k,1/\log k)$. Consider
a graph $G=(V,E)$ with at least $C n\gamma^{-1}$ edges
($k_0$, $c$, $C$ are absolute constants) and a random
linear unique game from the ``adversarial edges,  random constraints'' model.
Let $E_{\varepsilon}$ be the set of corrupted edges; and let $\{u_i\}$ be the optimal
SDP solution.  Then, with probability $1-o(1)$, at most $\gamma|E|$ edges in $E_{\varepsilon}$ are  $1/32$--short.
\end{theorem}
\begin{proof}
As mentioned in the beginning of this section the result we would like to prove
holds not only for cyclic groups but for arbitrary Abelian groups. In fact,
the proof is simpler if the group can be decomposed into a direct sum of many
cyclic groups of small size. So, in some sense, our plan is to represent the
cyclic group as a ``pseudo-direct sum'' of such groups.
Before proceeding to the proof, we note that all
transformations described below are only required for the proof and are not
performed by the approximation algorithm.

Let $S$ be a circle (or a one dimensional torus). We identify $S$ with the segment $[0,1]$ with
``glued'' ends 0 and 1. Define an embedding $\theta: [k]\to S$ as follows $\theta(i) = i/k$. We denote
the distance between two adjacent labels by $\tau = 1/k$. Below,
all ``$+$'' and ``$-$'' operations on elements of $S$ are ``modulo 1'' and all operations
on $[k]$ (and later $[T]$) are modulo $k$ (and $T$ respectively).

Recall, that in the ``adversarial edges, random constraints'' model of Linear Unique Games,
the adversary first chooses a set of edges $E_{\varepsilon}$, then the ``nature'' randomly
picks a shift $s_{uv}\in [k]$ for every $(u,v)\in E_{\varepsilon}$. We assume that the
nature picks ${s}_{uv}$ in the following way: it chooses a random $\tilde{s}_{uv}\in S$,
and then defines $s_{uv} = \rounddown{\tilde{s}_{uv}k}$. Then, clearly, $s_{uv}$ is
distributed uniformly in $[k]$.

If $M$ is a subset of $S$, then $M+s$ denotes the shift of $M$ by $s$.
We denote by $L(M)$ the set of labels covered by $M$:
$$L(M)= \{i\in [k]: \theta(i) \in M\}.$$
We let $\Psi_u(M) = \sum_{i\in L(M)} u_i$, and $\Psi'_u(M) = \Psi_u(M)/\|\Psi_u(M)\|$.

Fix $T = R^2$ (where $R\in \bbN$) to be a sufficiently large absolute constant (not depending on $k$ or
any other parameters).
We represent each $\tilde{s}\in S$ as
$$\tilde{s} = \sum_{p=1}^\infty \tilde{s}_p/T^p,$$
where every $\tilde{s}_p$ belongs to $[T]$, in other words,
$0.\tilde{s}_1\tilde{s}_2\tilde{s}_3\dots$ is a representation
of $\tilde{s}$ in the numeral system with base $T$. (This representation is not unique
for a set of measure $0$. If it is not unique, we pick the representation that ends in zeros.)
We define $D_p(\tilde{s}) = \tilde{s}_p$. Note, that if $\tilde{s}$ is uniformly distributed
random variable in $S$, then $D_p(\tilde{s})$ is uniformly distributed in $[T]$ and
all $D_p(\tilde s)$ are independent for different $p$. We are interested in the first
$P=\roundup{2\ln\ln k}$ digits of $\tilde{s}_{uv}$. We have chosen $P$ so that $T^{-P} = \log^{-\Theta(1)} k > \tau R \equiv R/k$, but $e^P\geq \ln^2 k$.

For every integer $p\in \{1,\dots, P\}$ and $d\in [T]$ define
$$M_{p,d} = \{\tilde{s}\in S: D_p(\tilde{s}) \in [d, d+R-1] \},$$
that is, $M_{p,d}$ is the set of numbers in $[0,1]$, whose $p$-th digit is in the range $[d, d+R-1]$. Note, that the set $\{\tilde{s}\in S: D_p(\tilde{s}) = d\}$ is a union
of segments of length $T^{-P} > \tau R\equiv R/k$ (here and below, all segments are left-closed and right-open).

\begin{lemma}\label{lem:segmunion}
Let $M\subset S$ be a union of segments of length at least $R\tau$ and let $\tilde{s}\in[0,1]$. Then,
$$(1-1/R) \;\mu (M)/\tau \leq |L(M)| \leq (1+1/R) \; \mu (M)/\tau.$$
and
$$|L(M)\symdiff L(M+\tilde{s})| \leq \min(|L(M)|, |L(M+\tilde{s})|) /R \times (2\tilde{s}/\tau + 2),$$
here $\symdiff$ denotes the symmetric difference of two sets; and $\mu$ is the
uniform measure on $[0,1]$.
\end{lemma}
\begin{proof}
Let $M$ be the disjoint union of segments $I\in \calI$ of length at least $R\tau$ each. Consider one of the segments $I\subset M$. This segment covers at least $\rounddown{\mu (I)/\tau}\geq (\mu (I)/\tau - 1)$
and at most $\roundup{\mu (I)/\tau} \leq (\mu (I)/\tau +1)$ points $\theta(i)$. In other words,
$\mu (I)/\tau  - 1 \leq |L(I)| \leq \mu (I)/\tau +1$. Observe that $1/R\cdot \mu(I)/\tau \geq 1$,
because $\mu(I)\geq R\tau$, thus $$(1-1/R)\mu (I)/\tau \leq |L(I)| \leq (1 + 1/R) \mu (I)/\tau.$$
Using equalities $|L(M)| = \sum_{I\in \calI} |L(I)|$  and $\mu(M) = \sum_{I\in \calI} \mu(I)$, we
get $(1-1/R) \mu (M)/\tau \leq  |L(M)| \leq (1+1/R) \mu (M)/\tau$.

To prove the second inequality, observe that
$$L(M)\symdiff L(M+\tilde{s})\subset \bigcup_{I\in \calI} (L(I)\symdiff L(I+\tilde{s})).$$
 For each $I\in \calI$,
$L(I)\symdiff L(I+\tilde{s}) = L(I\symdiff (I+\tilde{s}))$. The
set $I\symdiff (I+\tilde{s})$ is the union of two segments of length at most $\tilde{s}$,
thus $I\symdiff (I+\tilde{s})$ covers at most $2(\tilde{s}/\tau + 1)$ points $\theta(i)$ ($i\in [k]$).
Every interval covers at least $R$ points. Thus, the size of the family $\calI$ is at most $L(I)/R$.
Therefore, $|L(M)\symdiff L(M+\tilde{s})| \leq |L(M)|/R \times (2\tilde{s}/\tau + 2)$. The same
argument shows $|L(M)\symdiff L(M+\tilde{s})| \leq |L(M+\tilde{s})|/R \times (2\tilde{s}/\tau + 2)$.
\end{proof}

\begin{lemma}\label{lem:C.3}
Let $s_{uv} = \rounddown{k\tilde{s}_{uv}}$ ($s_{uv}\in [k]$, $\tilde{s}_{uv}\in S$), $p\in \{1,\dots, P\}$ and let $d=D_p(\tilde{s}_{uv})$ be the $p$-th digit of
$\tilde{s}_{uv}$. Then
$|L(M_{p,0})|, |L(M_{p,d})|\in [(1-1/R)/(R \tau), (1+1/R)/(R \tau)]$ and
$$|L(M_{p,0})\symdiff (L(M_{p,d}) - s_{uv})|\leq
8\min(|L(M_{p,0})|,|L(M_{p,d})|)/R .$$
\end{lemma}
\begin{proof}
We have $\mu(M_{p,d}) = R/T = 1/R$, thus by Lemma~\ref{lem:segmunion},
$|L(M_{p,0})|\in [(1-1/R)/(R \tau), (1+1/R)/(R \tau)]$.
Write,
\begin{eqnarray*}
|L(M_{p,0})\symdiff (L(M_{p,d}) - s_{uv})| &=& |(L(M_{p,0}) + s_{uv}) \symdiff L(M_{p,d})| =
|L(M_{p,0} + \theta(s_{uv})) \symdiff L(M_{p,d})|\\
&\leq&
|L(M_{p,0} + \theta(s_{uv})) \symdiff L(M_{p,0} + \tilde{s}_{uv})|+
|L(M_{p,0} + \tilde{s}_{uv}) \symdiff L(M_{p,d})|.
\end{eqnarray*}
Since $\tilde{s}_{uv} - \theta(s_{uv})\in [0,\tau]$, by Lemma~\ref{lem:segmunion},
$$|L(M_{p,0} + \theta(s_{uv})) \symdiff L(M_{p,0} + \tilde{s}_{uv})|\leq  |L(M_{p,0} + \theta(s_{uv}))|/R\cdot (2(\tilde{s}_{uv} - \theta(s_{uv}))/\tau + 2)\leq 4 |L(M_{p,0})|/R.$$
The $p$-th digit of $\tilde{s}_{uv}$ is $D_p(\tilde{s}_{uv}) = d$. Hence, the $p$-th digit of every number
$\tilde{s}$ in $M_{p,0} + \tilde{s}_{uv}$ is in the range $[d, (R-1) + d + 1] = [d, R + d]$. Moreover,
all numbers with $p$-th digit in the interval $[d+1, (R-1) + d]$ are covered by
$M_{p,0} + \tilde{s}_{uv}$. Thus,
$(M_{p,0} + \tilde{s}_{uv})\symdiff M_{p,d} \subset \{\tilde{s}: D_p(\tilde{s})\in \{d,d+R\}\}.$
The measure of the set $\{\tilde{s}: D_p(\tilde{s})\in \{d,d+R\}\}$ is $2/T$. It is
a union of segments of lengths $T^{-p}\geq \tau R$. By Lemma~\ref{lem:segmunion},
$$L((M_{p,0} + \tilde{s}_{uv})\symdiff M_{p,d})\subset L(\{\tilde{s}: D_p(\tilde{s})\in \{d,d+R\}\}) \leq (1+1/R)\cdot 2/(T\tau) \leq 4L(M_{p,0})/R.$$
\end{proof}

\begin{lemma}~\label{lem:sizeM}
Let $s_{uv} = \rounddown{k\tilde{s}_{uv}}$ ($s_{uv}\in [k]$, $\tilde{s}_{uv}\in S$),
$p\in \{1,\dots, P\}$, and let $d=D_p(\tilde{s}_{uv})$ be the $p$-th digit of
$\tilde{s}_{uv}$. Suppose that an edge $(u,v)\in E_{\varepsilon}$ is 1/32--short, then
\begin{equation}\label{eq:psimpd1}
\|\Psi_u(M_{p,0}) - \Psi_v(M_{p,d})\|^2 \leq 1/8 \;
\min (\|\Psi_u(M_{p,0})\|^2, \|\Psi_v(M_{p,D_p(\tilde{s})})\|^2);
\end{equation}
and
\begin{equation}\label{eq:psimpd2}
\|\Psi'_u(M_{p,0}) - \Psi'_v(M_{p,d})\|^2 \leq 1/8.
\end{equation}
\end{lemma}
\begin{proof}
Write,
\begin{eqnarray*}
\|\Psi_u(M_{p,0}) - \Psi_v(M_{p,d})\|^2 &=&
\big\|\sum_{i\in L(M_{p,0})} u_i - \sum_{i\in L(M_{p,d})} v_i\big\|^2
=
\Big\|\sum_{i\in L(M_{p,0}) \cap (L(M_{p,d}) - s_{uv})} (u_i - v_{i+s_{uv}}) + \\
&\phantom{=}&\sum_{i\in L(M_{p,0}) \setminus (L(M_{p,d}) - s_{uv})} u_i -
\sum_{i\in  (L(M_{p,d}) - s_{uv}) \setminus L(M_{p,0})} v_{i+s_{uv}}\Big\|^2.
\end{eqnarray*}
Using $\ell_2^2$ triangle inequalities $\langle u_i - v_{i+s_{uv}}, u_j - v_{j+s_{uv}}\rangle =
-\langle u_i, v_{j+s_{uv}}\rangle - \langle v_{i+s_{uv}}, u_j\rangle \leq 0$, $\langle u_i, -v_{j+s_{uv}}\rangle \leq 0$
(for $i\neq j$) and then
Lemma~\ref{lem:C.3}, we get (for sufficiently large $R$)
\begin{eqnarray*}
\|\Psi_u(M_{p,0}) - \Psi_v(M_{p,d})\|^2 &\leq&
\sum_{i\in L(M_{p,0}) \cap (L(M_{p,d}) - s_{uv})} \|u_i - v_{i+s_{uv}}\|^2 +
|L(M_{p,0}) \symdiff (L(M_{p,d}) - s_{uv})|\\
&\leq& 1/16 \; |L(M_{p,0})|   + 8|L(M_{p,0})|/R \leq 1/8\; |L(M_{p,0})| = \|\Psi_u(M_{p,0})\|^2/8.
\end{eqnarray*}
Similarly,
$\|\Psi_u(M_{p,0}) - \Psi_v(M_{p,d})\|^2 \leq \|\Psi_u(M_{p,d})\|^2/8$.

Inequality~(\ref{eq:psimpd2}) immediately follows from inequality~(\ref{eq:psimpd1}): let $\psi_u = \Psi_u(M_{p,0})$,
$\psi_v = \Psi_v(M_{p,d})$ and assume $\|\psi_u\|\leq \|\psi_v\|$, then
$\|\;\psi_u/\|\psi_u\| - \psi_v / \|\psi_u\|\; \|^2 \leq  1/8$. Vector
$\psi_u/\|\psi_u\|$ has length 1, and vector $\psi_v / \|\psi_u\|$ has length at least 1,
hence $\|\;\psi_u/\|\psi_u\| - \psi_v / \|\psi_v\|\; \|^2 \leq \|\;\psi_u/\|\psi_u\| - \psi_v / \|\psi_u\|\; \|^2 \leq 1/8 $.
\end{proof}

Observe, that vectors $\Psi'_v(M_{p,d'})$, $\Psi'_v(M_{p,d''})$ are orthogonal if $|d''-d'| > R$, and thus $$\|\Psi'_u(M_{p,d'}) - \Psi'_u(M_{p,d''})\|^2 = 2.$$

We now proceed the same way as in the proof of Theorem~\ref{thm:longedges}.
We embed $O(n\log\log k)$ vectors $\Psi'_u(M_{p,d})$
($u\in V$, $p\in P$, $d\in T$) in a net $N$ of size $O(1)$
using a randomized mapping $\varphi$ (see Lemma~\ref{lem:JL2}), so
that for some small absolute constant $\beta$ and every
$u,v\in  V$; $d',d'' \in T$; and $p\in \{1,\dots,P\}$,
\begin{multline*}
\Pr ((1+\beta)^{-1} \|\Psi'_u(M_{p,d'})-\Psi'_v(M_{p,d''})\|^2 - \beta
\leq \|\Phi(u,p,d') - \Phi(v,p,d'')\|^2 \\
\leq (1+\beta)\|\Psi'_u(M_{p,d'})-\Psi'_v(M_{p,d''})\|^2 + \beta) \geq 1 - \beta/T^{2},
\end{multline*}
where $\Phi(u,p,d) = \varphi(\Psi'_u(M_{p,d}))$. We say that a pair
$(u,p)\in V\times \{1,\dots,P\}$ is good if the following inequality holds:
$\|\Phi(u,p,d')- \Phi(u,p,d'') \|^2 \geq 2 - 3\beta$
for all $d',d''\in [T]$ such that $|d'' - d'| > R$. By Lemma~\ref{lem:JL2},
a pair $(u,p)$ is good with probability at least $1-\beta$.
Then, for a fixed $1/32$--short edge $(u,v)\in E_{\varepsilon}$, the expected fraction of
$p$'s for which both pairs $(u,p)$ and $(v,p)$ are good and
\begin{equation}\label{eq:phi18}
\|\Phi(u,p,0) - \Phi(v,p,D_p(\tilde{s}_{uv}))\|^2 \leq 1/4
\end{equation}
is at least $1-3\beta$.

Assume that $\gamma |E| = \gamma\varepsilon^{-1} |E_{\varepsilon}|$ edges in $E_{\varepsilon}$
are $1/32$--short. We say that an edge $(u,v)\in E_{\varepsilon}$ is good with respect to the set
$\{\Phi(u',p,d)_{u,p,d}\}$ if the following statement holds:
 ``for at least $(1-6\beta)$ fraction of $p$'s in $\{1,\dots,P\}$, the pairs
$(u,p)$ and $(v,p)$ are good, and inequality~(\ref{eq:phi18}) holds''. By the Markov inequality,
there exists a realization of random variables $\Phi(u,p,d)$ (random with respect to a random
embedding in the net $N$) such that for at least $\gamma/2$ fraction of edges $(u,v)\in E_{\varepsilon}$
the previous statement holds.

Thus we have shown that for every linear unique game
with $\gamma|E|$   $1/32$--short edges there always exists
a witness---a collection of vectors $\{\Phi(u,p,d)\}\subset N$, such that at least $\gamma|E|/2$
edges in $E_{\varepsilon}$ are good with respect to this collection (Remark: so far we have not
used that the instance is semi-random). Now, we will prove
that for a fixed witness, the probability that $\gamma|E|/2$
 edges in $E_{\varepsilon}$ is good in a semi-random unique game is exponentially
small.

Fix an edge $(u,v)\in E_{\varepsilon}$ and compute the probability that it is good with respect to
a fixed witness $\{\Phi(u,p,d)\}\subset N$. The probability that
\begin{equation}\label{Phi14}
\|\Phi(u,p,0) - \Phi(v,p,D_p(\tilde{s}_{uv}))\|^2 \leq 1/4
\end{equation}
for a random $(\tilde{s}_{uv})$ is at most $1/R$ if pairs $(u,p)$ and $(v,p)$ are good,
since among every $R$ values $d \in\{d_0,d_0+R, d_0 + 2R,\cdots\}$ there is at most one $d$ satisfying
$$\|\Phi(u,p,0) - \Phi(v,p,d)\|^2 \leq 1/4.$$
(Recall, that $\|\Phi(v,p,d') - \Phi(v,p,d'')\|^2 \geq 2 - 3\beta$ if $|d'-d''| > R$).
By the Chernoff bound the probability that for $(1-6\beta)$ fraction of $p$'s
the inequality~(\ref{Phi14}) is satisfied and $(u,p)$, $(v,p)$ are good
is at most
$$e^{-P\ln ((1-6\beta)R)} \leq e^{-P} \leq \ln^{-2}k.$$

Hence, the expected number of good edges in $E_{\varepsilon}$ is at most
$|E_{\varepsilon}|/\ln^2k = \varepsilon|E|/\ln^2k $, and
the probability that $\gamma|E|/2$ edges in $E_{\varepsilon}$ are good is at most
$$\exp\left(-\gamma|E|/2 \cdot \ln \bigl(\gamma\varepsilon^{-1}\ln^2 k /2\bigr)\right) \leq
\exp(-Cn/2 \cdot \ln (1/2\,\ln k)).$$

The total number of possible witnesses $\{\Phi(u,p,d)\}\subset N$ is
$\exp(O(n \log\log k))$. So by the union bound (for sufficiently large
absolute constant $C$) with probability $1-\exp(-n)=1-o(1)$,
less  than $\gamma |E|$ edges in $E_{\varepsilon}$ are $1/32$--short.
\end{proof}

\section{Random Initial Constraints}
In this section, we consider the model, in which the initial set of constraints is chosen at random and
other steps are controlled by the adversary.
Specifically, in this model the adversary chooses the constraint graph $G = (V, E)$ and a ``planted solution''
$\set{x_u}$. Then for every edge $(u,v)\in E$, she randomly chooses a permutation (constraint)
$\pi_{uv}$ such that $\pi_{uv}(x_u) = x_v$ (each of $(k-1)!$ possible permutations is chosen with the same probability $1/(k-1)!$;
choices for different edges are independent).
Then the adversary chooses an arbitrary set $E_\eps$ of edges of size at most $\eps |E|$ and (adversarially) changes the corresponding
constraints: replaces constraint $\pi_{uv}$ with a constraint $\pi_{uv}'$ for $(u,v)\in E_\eps$.
Note that the obtained semi-random instance is $1-\eps$ satisfiable since the ``planted solution'' $x_u$ satisfies constraints for
edges in $E\setminus E_\eps$.
The analysis of this model is much simpler than the analysis of the other two models that we study.
\begin{theorem}\label{thm:main3}
There exists a polynomial-time algorithm that given
$k \in \bbN$ ($k\geq k_0$), $\varepsilon \in (0,1)$, $\eta \in (c\log k/\sqrt{k},1)$
and a semi-random instance of unique games from the ``random initial instance''
model on graph $G=(V,E)$ with at least $C \eta^{-1} n\log k$ edges
finds a solution of value $1-O(\varepsilon + \eta/\log k)$
with probability $1-o(1)$ (where $c$, $C$ and $k_0$ are some absolute constants).
\end{theorem}
\begin{proof}
We solve the standard SDP relaxation for the problem. Then we use a very
simple rounding procedure. For every vertex $u$, if $\|u_i\|^2 > 1/2$ for some label $i$,
we label $u$ with $i$; otherwise, we label $u$ with an arbitrary label
(since for every two labels $i_1\neq i_2$, $\|u_{i_1}\|^2 + \|u_{i_2}\|^2 \leq 1$,
we do not label any vertex with two labels).

We now show that our labeling satisfies a $1-O(\eps + \eta/\log k)$ fraction of constraints w.h.p.
Without loss of generality, we assume that the planted solution is $x_u = 0$ for every $u\in V$.
For every $t \in (1/2; 3/4)$, let $S_t = \set{u: \|u_0\|^2 \geq t}$.
Let $t_0$ be the value of $t$ that minimizes the size of the cut between
$S_t$ and $V\setminus S_{t}$; let $S = S_{t_0}$.
Note that if $u\in S$ then we label vertex $u$ with $0$.
Therefore, our labeling satisfies all constraints $\pi_{uv}$ for edges $(u,v)$ within $S$
(but not necessarily constraints $\pi_{uv}'$). We conservatively assume that constraints for all
edges from $S$ to $V\setminus S$ and edges within $V\setminus S$ are not satisfied.
We now estimate their number.
First, we bound the number of edges leaving $S$.
Note that since the unique game instance is $1-\eps$ satisfiable the cost
of the SDP solution is at most $\eps$.
In particular,
$$\frac{1}{2}\sum_{(u,v)\in E} \Bigl|\|u_0\|^2 - \|v_0\|^2\Bigr| \leq \frac{1}{2}\sum_{(u,v)\in E} \|u_0 - v_0\|^2 \leq \eps.$$
On the other hand, if we choose $t$ uniformly at random from $(1/2; 3/4)$, then the
probability that $u\in S_t$ and $v\notin S_t$ or $u\notin S_t$ and $v\in S_t$ is at most
$4\Bigl|\|u_0\|^2 - \|v_0\|^2\Bigr|$ for every $(u,v)\in E$. Therefore, the expected
size of the cut between $S_t$ and $V\setminus S_t$ is at most $8 \eps|E|$.
Hence the size of the cut between $S$ and $V\setminus S$ is at most $8 \eps|E|$.

Now we estimate the number of edges within $V\setminus S$.
We consider a new instance of Unique Games on $G$ with the label
set $\set{1,\dots, k-1}$ and constraints $\pi_{uv}^* = \pi_{uv}|_{\set{1,\dots, k-1}}$
(the restriction of $\pi_{uv}$ to the set $\set{1,\dots, k-1}$).
Note that each $\pi_{uv}^*$ is a permutation of $\set{1,\dots, k-1}$ since
$\pi_{uv}(0) = 0$. Moreover, each $\pi_{uv}^*$ is a random permutation
uniformly chosen among all permutations on $\set{1,\dots, k-1}$.
For each vertex $u$ and label $i\in\{1,\dots, k-1\}$, we define a vector $u_i^*$ as follows.
If $u\notin S$, we let $u_i^* = u_i$. Otherwise, we let $u_1^* = e$ and $u_i^* = 0$ for $i > 1$,
where $e$ is a fixed unit vector orthogonal to all vectors $v_j$ in the SDP solution.

We say that $\{u_i^*\}$ is a relaxed SDP solution if it satisfies all SDP conditions except possibly
for the condition that $\sum_{i} \|u_i^*\|^2 = 1$ for every vertex $u$.
We require instead that $1/4 \leq \sum_{i} \|u_i^*\|^2  \leq 1$.
Note that the set of vectors $\{u_i^*\}$ is a relaxed SDP solution since
for every $u\notin S$, $\sum_{i=1}^{k-1} \|u_i^*\|^2 = 1 - \|u_0\|^2 \geq 1 - t_0 \geq 1/4$;
for every $u\in S$, $\sum_{i=1}^{k-1} \|u_i^*\|^2 = \|e\|^2 = 1$.

We now use a slightly modified version of Theorem~\ref{thm:longedges}.
\begin{lemma}\label{lem:longedgesmod}
Consider a unique game on a graph $G=(V,E)$ with at least $C \eta^{-1} n\log k$ edges
with random set of constraints $\pi_{uv}^*$, where $\eta \in (c\log k/\sqrt{k}, 1)$
(where $c$ and $C$ are some absolute constants).
Let $\{u_i^*\}$ be a relaxed SDP solution. Then, there are at most $O(\eta/\log k)|E|$ $1/64$--short edges
with probability $1-o(1)$.
\end{lemma}
The proof of the lemma almost exactly repeats the proof of Theorem~\ref{thm:longedges} for $\eps = 1$
(we only need to change inequalities (\ref{eq:f-sum1}) and (\ref{eq:f-sum2}) slightly).

We apply this lemma to the solution $\{u_i^*\}$. We get that there are at most  $O(\eta/\log k)|E|$
$1/64$--short edges. In particular, there are at most $O(\eta/\log k)|E|$ $1/64$--short edges in
$E(V\setminus S)$. Thus
\begin{align*}
\frac{1}{2}\sum_{(u,v) \in E}\sum_{i=1}^{k-1} \|u_i - v_{\pi_{uv}(i)}\|^2 &\geq
\frac{1}{2}\sum_{(u,v) \in E(V\setminus S)}\sum_{i=1}^{k-1} \|u_i - v_{\pi_{uv}(i)}\|^2 = \frac{1}{2}\sum_{(u,v) \in E(V\setminus S)}\sum_{i=1}^{k-1} \|u_i^* - v_{\pi^*_{uv}(i)}^*\|^2 \\
&\geq \frac{|E(V\setminus S)|}{64} - O(\eta/\log k)|E|.
\end{align*}
However, the left hand side is at most $\eps |E|$. Therefore, $E(V\setminus S) = O(\eps + \eta/\log k)|E|$.

We conclude that the solution that our algorithm finds satisfies a $1 - O(\eps + \eta/\log k)$ fraction of
constraints $\pi_{uv}$. Since there are at most $\eps|E|$ corrupted constraints, the solution also
satisfies a $1 - O(\eps + \eta/\log k) - \eps = 1 - O(\eps + \eta/\log k)$ fraction of corrupted constraints.
\end{proof}

\section{Distinguishing Between Semi-Random Unsatisfiable Games and Almost Satisfiable Games}
In this section, we study the following question, Is it possible to distinguish between $(1-\eps)$ satisfiable
semi-random games and $(1- \delta)$ satisfiable (non-random) games if $\delta \ll \eps$?
This question is interesting only in the model where the corrupted constraints are chosen at random (i.e. step 4 is random),
since in the other two semi-random models (when the initial constraints are random, and
when the set of corrupted edges $E_\eps$ is random), the semi-random instance can be $1-\delta$
satisfiable, therefore, the answer is trivially negative.

Specifically, we consider the following model. The adversary chooses a constraint graph $G$ and
a set of constraints $\pi_{uv}$. We do not require that this instance is completely satisfiable.
Then she chooses a set of edges $E_{\eps}$ of size $\eps|E|$. She replaces
constraints for edges in $E_{\eps}$ with \textit{random} constraints (each constraint is chosen uniformly
at random among all $k!$ possible constraints).

We show that such semi-random instance can be distinguished w.h.p. from a $(1-\delta)$ satisfiable instance
if $\delta < c \eps$ (where $c$ is an absolute constant) if $|E| \geq C n \max(\eps^{-1}, \log k)$.
To this end, we consider the standard SDP relaxation for Unique Games.
We prove that the SDP value of a semi-random instance is at least $c\eps$;
whereas, of course, the SDP value of a $(1-\delta)$ satisfiable instance is at most $\delta$.

\begin{theorem}
Let $k \in \bbN$ ($k\geq k_0$) and $\varepsilon \in (0,1]$. Consider a graph $G$ with at least
$Cn\max(\varepsilon^{-1},\log k)$ edges, and a semi-random unique games instance $\cal I$
on $G$ with $\eps |E|$ randomly corrupted constraints ($k_0$ and $C$ are absolute constants).
Then the SDP value of $\cal I$ is at least $\varepsilon/32$ with probability $1-o(1)$.
\end{theorem}
\begin{proof}
We apply Theorem~\ref{thm:longedges} to our instance of Unique Games, with $\gamma = \min(\varepsilon/2, 1/(2\log k))$.
We get that at least half of all edges in $E_{\eps}$ are $1/16$--long w.h.p.
The contribution of these edges to the sum $\frac{1}{2}\sum_{(u,v)}\sum_i \|u_i - v_{\pi_{uv}(i)}\|^2$
in the SDP objective function is at least $1/16 \times (|E_{\eps}|/2) = \eps|E|/32$.
Therefore, the SDP value is at least $\eps/32$.
\end{proof}

%%%%%%%%%%%%%%%%%%%%%%%%%%%%%%%%%%%%%%%%%%%%%%%%%%%%%%%%%%%%%%%%%%%%%%%%


\begin{thebibliography}{1W}
\bibitem{Alekhnovich} M.~Alekhnovich.
\newblock\textit{More on Average Case vs Approximation Complexity.}
\newblock In Proceedings of the 44st IEEE Symposium on Foundations of Computer Science: pp,~298--307, 2003.

\bibitem{AEH} G. Andersson, L. Engebretsen, and J. H{\aa}stad.
\newblock\textit{A new way to use semidefinite programming with applications to linear equations mod p,} Journal of Algorithms, Vol.~39, 2001, pp. 162--204.

\bibitem{ABS} S. Arora, B. Barak, and D. Steurer.
\newblock\textit{Subexponential Algorithms for Unique Games and Related problems.}
\newblock{In Proceedings of the 51st IEEE Symposium on Foundations of Computer Science, 2010.}

\bibitem{AIMS} S. Arora, R. Impagliazzo, W. Matthews, and D. Stuerer.
\newblock\textit{Improved algorithms for unique games via divide and conquer,}
\newblock In Electronic Colloquium on Computational Complexity, TR10-041, 2010.

\bibitem{AKK}
S. Arora, S. Khot, A. Kolla, D. Steurer, M. Tulsiani, and N. Vishnoi.
\newblock\textit{Unique games on expanding constraint graphs are easy.}
\newblock In Proceedings of the 40th ACM Symposium on Theory of Computing, pp. 21--28, 2008.


\bibitem{BHHRRS} B.~Barak, M. Hardt, I. Haviv, A. Rao, O. Regev and D. Steurer.
\newblock\textit{Rounding Parallel Repetitions of Unique Games},
\newblock In Proceedings of the 49th IEEE Symposium on Foundations of Computer Science, pp.~374--383, 2008.

\bibitem{BS} A.~Blum and J.~Spencer.
\newblock\textit{Coloring Random and Semi-Random k-Colorable Graphs,}
\newblock J. Algorithms, vol.~19, no. 2, pp.~204--234, 1995.

\bibitem{CMM1}
M.~Charikar, K.~Makarychev, and Y.~Makarychev.
\newblock \textit{Near-Optimal Algorithms for Unique Games.}
\newblock In Proceedings of the 38th ACM Symposium on Theory of Computing,
pp.~205--214, 2006.
%%
%%
\bibitem{CMM2}
E.~Chlamtac, K.~Makarychev, and Y.~Makarychev.
\newblock\textit{How to Play Unique Games Using Embeddings.}
\newblock In Proceedings of the 47th IEEE Symposium on Foundations of Computer Science, pp.~687--696, 2006.
\bibitem{GT}
A.~Gupta and K.~Talwar.
\newblock \textit{Approximating Unique Games.}
\newblock In Proceedings of the 17th ACM-SIAM Symposium on Discrete Algorithms,
pp.~99--106, 2006.
%%

\bibitem{Feige} U. Feige.
\newblock\textit{Relations Between Average Case Complexity and Approximation Complexity}
\newblock In Proceedings of the 34th Annual ACM Symposium on Theory of Computing, pp.~534--543, 2002.

\bibitem{FKil} U. Feige and J. Kilian,
\newblock\textit{Heuristics for Semirandom Graph Problems},
\newblock Journal of Computing and System Sciences, vol. 63, pp.~639--671, 2001.

\bibitem{FKra} U.~Feige and R.~Krauthgamer.
\newblock\textit{Finding and Certifying a Large Hidden Clique in a Semi-Random Graph}
\newblock Random Structures and Algorithms, vol.~16(2), pp.~195--208, 2000.

\bibitem{GMR}
V. Guruswami, R. Manokaran, and P. Raghavendra.
\newblock \textit{Beating the Random Ordering is Hard: Inapproximability of Maximum Acyclic Subgraph}.
\newblock In Proceedings of the 49th IEEE Symposium on Foundations of Computer Science, pp.~573--582, 2008.

\bibitem{GR08}V.~Guruswami and P.~Raghavendra.
\newblock \textit{Constraint satisfaction over a non-boolean domain: Approximation
  algorithms and unique-games hardness.}
\newblock In Proceedings of APPROX-RANDOM, 77--90, 2008.

\bibitem{Jerrum} M. Jerrum.
\newblock \textit{Large Cliques Elude the Metropolis Process,}
\newblock Random Structures and Algorithm, vol. 3 (4), pp.~347--359, 1992.

\bibitem{Kho02} S. Khot.
\newblock \textit{On the power of unique 2-prover 1-round games}.
\newblock In Proceedings of the 34th ACM Symposium on Theory of Computing,
pp.~767--775, 2002.
%%

\bibitem{KKMO05} S.~Khot, G.~Kindler, E.~Mossel, and R.~O'Donnell.
\newblock{\textit{Optimal inapproximability results for MAX-CUT and other 2-variable CSPs?}}
\newblock ECCC Report TR05-101, 2005.
%%
\bibitem{KR03} S. Khot and O. Regev.
\newblock{\textit{Vertex cover might be hard to approximate to within $2-\varepsilon$.}}
\newblock In Proceedings of the 18th IEEE Annual Conference on Computational Complexity, 2003.

\bibitem{KV} S. Khot and N. Vishnoi.
\newblock\textit{The Unique Games Conjecture, Integrality Gap for Cut Problems and Embeddability of Negative Type Metrics into $\ell_1$,}
\newblock Foundations of Computer Science, pp.~53--62, 2005.

\bibitem{Kolla} A.~Kolla.
\newblock{\textit{Spectral Algorithms for Unique Games}}
\newblock In Proceedings of the Conference on Computational Complexity, pp.~122--130, 2010.

\bibitem{MM} K.~Makarychev and Y.~Makarychev.
\newblock{\textit{How to Play Unique Games on Expanders,}}
\newblock In Proceedings of the eighth Workshop on Approximation and Online Algorithms, 2010.

\bibitem{Rag} P.~Raghavendra,
\newblock\textit{Optimal algorithms and inapproximability results for every CSP?},
\newblock Proceedings of the 40th Annual ACM Symposium on Theory of Computing, pp.~245--254, 2008.

\bibitem{RS} P.~Raghavendra, D.~Steurer,
\newblock\textit{How to Round Any CSP,}
\newblock
In Proceedings of the 50th IEEE Symposium on Foundations of
Computer Science, pp.~586--594, 2009.

\bibitem{ST06} A.~Samorodnitsky and L.~Trevisan.
\newblock \textit{Gowers uniformity, influence of variables, and PCPs.}
\newblock In Proceedings of the 38th annual ACM symposium on Theory of computing, pp.~11--20, 2006.

\bibitem{Tre05}
L. Trevisan.
\newblock \textit{Approximation Algorithms for Unique Games.}
\newblock  In Proceedings of the 46th IEEE Symposium on Foundations of
Computer Science, pp.~197--205, 2005.

\end{thebibliography}
\end{document}